\documentclass[a4paper,USenglish,cleveref, thm-restate]{lipics-v2021}

\usepackage{booktabs}
\usepackage{xcolor}
\usepackage{cite}
\usepackage{algorithm}
\usepackage[noend]{algpseudocode}
\usepackage{algorithmicx}

\algdef{SE}[SUBALG]{Indent}{EndIndent}{}{\algorithmicend\ }%
\algtext*{Indent}
\algtext*{EndIndent}

\newcommand{\ignore}[1]{}

\newtheorem{assumption}{Assumption}

\title{VBR: Version Based Reclamation} 

\author{Gali Sheffi}{Department of Computer Science, Technion, Haifa, Israel}{galish@cs.technion.ac.il}{}{}

\author{Maurice Herlihy}{Department of Computer Science, Brown University, Providence, USA}{mph@cs.brown.edu}{}{}

\author{Erez Petrank}{Department of Computer Science, Technion, Haifa, Israel}{erez@cs.technion.ac.il}{}{}

\authorrunning{G. Sheffi, M. Herlihy and E. Petrank} 



\hideLIPIcs  

\EventEditors{John Q. Open and Joan R. Access}
\EventNoEds{2}
\EventLongTitle{42nd Conference on Very Important Topics (CVIT 2016)}
\EventShortTitle{CVIT 2016}
\EventAcronym{CVIT}
\EventYear{2016}
\EventDate{December 24--27, 2016}
\EventLocation{Little Whinging, United Kingdom}
\EventLogo{}
\SeriesVolume{42}
\ArticleNo{23}

\begin{document}

\maketitle

\begin{abstract}
Safe lock-free memory reclamation is a difficult problem. 
Existing solutions follow three basic methods (or their combinations): epoch based reclamation, hazard pointers, and optimistic reclamation. 
Epoch-based methods are fast, but do not guarantee lock-freedom. Hazard pointer solutions are lock-free but typically do not provide high performance. 
Optimistic methods are lock-free and fast, but previous optimistic methods did not go all the way. While reads were executed optimistically, writes were protected by hazard pointers. In this work we present a new reclamation scheme called {\em version based reclamation} (VBR), which provides a full optimistic solution to lock-free memory reclamation, obtaining lock-freedom and high efficiency. Speculative execution is known as a fundamental tool for improving performance in various areas of computer science, and indeed evaluation with a lock-free linked-list, hash-table and skip-list shows that VBR  outperforms state-of-the-art existing solutions. 
\end{abstract}

\section{Introduction} \label{sec-intro}

Lock-freedom guarantees eventual system-wide progress, regardless of the behavior of the executing threads.
Achieving this desirable progress guarantee in practice requires a lock-free memory reclamation mechanism.
Otherwise, available memory space may be exhausted and the executing threads may be indefinitely blocked while attempting to allocate, foiling any progress guarantee. 
Automatic garbage collection could solve this problem for high-level managed languages, but while some efforts have been put into designing a garbage collector that supports lock-free executions~\cite{herlihy1991lock,hudson2001sapphire,pizlo2007stopless,pizlo2008study,pizlo2010schism}, a lock-free garbage collector for the entire heap is not available in the literature. Consequently, lock-free implementations must use manual memory reclamation schemes.

Manual reclamation methods rely on \textit{retire} invocations by the program, announcing that a certain object has been unlinked from a data-structure. 
After an object is retired, the task of the memory reclamation mechanism is to decide when it is safe to reclaim it, making its memory space available for reuse in future allocations.
The memory reclamation mechanism ensures that an object is not freed by one thread, while another thread is still using it.
Accessing a memory address which is no longer valid may result in unexpected and undesirable program behavior. Conservative reclamation methods make sure that no thread accesses reclaimed space, while optimistic methods allow threads to speculatively access reclaimed memory, taking care to preserve correctness nevertheless. 

Conservative manual reclamation schemes can be roughly classified as either \textit{epoch--based} or \textit{pointer--based}. 
In epoch--based reclamation (EBR) schemes~\cite{harris2001pragmatic,hart2007performance,mckenney1998read,brown2015reclaiming}, all threads share a global epoch counter, which is incremented periodically. 
Additionally, the threads share an announcements array, in which they record the last seen epoch. 
During reclamation, only objects that had been retired before the earliest recorded epoch are reclaimed.
These schemes are often fast, but they are not robust. I.e., a stalled thread may prevent the reclamation of an unbounded number of retired objects. At a worst-case scenario, these objects will consume too much memory space resulting in blocking all new allocations and consequently, foiling system-wide progress.

Pointer--based reclamation methods~\cite{herlihy2005nonblocking,michael2004hazard,michael2002safe,solomon2021efficiently} allow threads to protect specific objects. Namely, before accessing an object, a thread can announce its access in order to prevent this object from being reclaimed. While pointer-based methods can guarantee robustness (and consequentially, lock-freedom), they incur significant time overheads because they need to protect each dereference to shared memory, and issue an expensive memory synchronization fence to make sure  the protection is visible to all threads before accessing the derefeneced object. 
Furthermore, pointer-based schemes are not applicable to many concurrent data-structures (e.g., to Harris's linked-list~\cite{harris2001pragmatic}).

Hybrid schemes that enjoy stronger progress guarantees and smaller time overheads have been proposed. Some epoch-based algorithms try to minimize the chance a non-cooperative thread will block the entire system, by allowing reclamation of objects whose life cycle do not overlap with the activity of the non-responsive thread, e.g., HE (Hazard Eras~\cite{ramalhete2017brief}) and IBR (Interval-Based Reclamation~\cite{wen2018interval}).
However, while these algorithms are effectively non-blocking in most practical scenarios, a halted thread can still prevent the reclamation of a large space that relates to the size of the heap.  
There exists a wait-free variant of the Hazard Eras algorithm~\cite{nikolaev2020universal} which provides a better guarantee but is naturally slower. 
Another hybrid approach, called \textit{PEBR}~\cite{kang2020marriage}, obtains lock-freedom at a lower cost, but relies on the elimination of the costly memory fences using the mechanism of~\cite{dice2016fast}, which in turn relies on hardware modifications or on undocumented operating systems assumptions that might not hold in the future. 
Another hybrid of pointer- and epoch-based reclamation is the DEBRA+ and the subsequent  NBR~\cite{brown2015reclaiming,singh2021nbr} reclamation schemes. In order to deal with stuck threads in the epoch-based reclamation, these schemes signal non-cooperative threads and prevent them from postponing the reclamation procedure. While DEBRA+ and NBR are fast, their lock-free property relies on the system's lock-free implementation of signaling. This assumption is not currently available in most existing operating systems, but it may become available in the future. 

The first optimistic approach to lock-free memory reclamation was the \textit{Optimistic Access} scheme~\cite{cohen2015efficient} (also denoted as OA), speculatively allowing reads from retired objects, but protecting writes from modifying retired objects through the use of conservative pointer-based reclamation.  
After each read, designated per-thread flags signify whether reclamation took place. This allows the reading thread to avoid using stale data loaded from reclaimed space. 
Subsequent work~\cite{cohen2015automatic,Cohen18} increased automation and applicability. 
While the optimistic access reclamation scheme initiated speculative memory reclamation, its mechanism only allowed speculative read instructions. Write instructions were still applied conservatively, using hazard pointers (also denoted as HP) protection to avoid writes on reclaimed space, limiting the benefits of speculative execution. 

In this paper we present \textit{VBR}, a novel optimistic memory reclamation scheme that employs versions to allow full speculative execution in which both read and write instructions are allowed to access reclaimed space. Both accesses are guarded from affecting program semantics by extending the object with versions for each mutable field. Write operations to shared memory will fail due to modified field versions. In addition, loads of stale values will be ignored because the versions will signify a possible memory reclamation. 
VBR is fully optimistic. Unlike OA~\cite{cohen2015efficient,cohen2015automatic,Cohen18}, writes are also speculative and do not require costly fences. 
Memory fences are used infrequently, upon updating the global epoch. 
VBR provides full lock-freedom and it does not allow a non-cooperative thread to stall the reclamation process. 
VBR does not rely on hardware or operating system assumptions (or modifications), except for the existence of a wide (double-word) compare and swap instruction, which is available on most existing architectures (e.g., x86). 

The proposed VBR can reuse any retired object immediately after it is retired without jeopardizing correctness, which makes it highly space efficient. However, VBR does encounter an issue that pops up in several other schemes~\cite{michael2004hazard,cohen2015efficient,cohen2015automatic,brown2015reclaiming,singh2021nbr}. 
The memory manager sometimes causes a read or a write instruction to ``fail'' due to memory reclamation validation test. This failure is not part of the original program control flow and thus, an adequate handling of such failure should be added. In~\cite{michael2004hazard} Michael proposes informally to "skip over the hazards and follow the path of the target algorithm when conflict is detected, i.e., try again, backoff, exit loop, etc.". Indeed in many known lock-free data structures~\cite{harris2001pragmatic,natarajan2014fast,shalev2006split,ellen2010non,michael2002high,linden2013skiplist,michael1996simple,fraser2004practical}, it is easy to handle such failures adequately. In a similar manner, all failures in our VBR scheme are easily handled for these lock-free data structures. A first rigorous treatment of these failures was provided in~\cite{cohen2015efficient} for data structures that are written in the normalized form of lock-free data structures~\cite{timnat2014practical}. Subsequently, a weaker version of normalized concurrent data structures was presented in~\cite{singh2021nbr}, where this problem is somewhat more severe, as signals may occur at an arbitrary point in the program flow. We follow this line of papers by rigorously addressing the case in which a read or a write validation fails. 


We have implemented VBR on a linked-list, a skip-list, and a hash table and evaluated it against epoch-based reclamation, hazard pointers, hazard eras, interval-based reclamation, and no reclamation at all. As expected, speculative execution outperforms conservative approaches and so VBR yields both lock-freedom as well as high performance. 

The paper is organized as follows. In Section~\ref{sec-overview} we provide an overview of the VBR scheme. In Section~\ref{sec-settings}, we describe our shared memory model and specify some assumptions a data structure must satisfy in order to be correctly integrated with our reclamation mechanism. We describe the VBR scheme integration in Section~\ref{sec-algorithm}. Experiments appear in Section~\ref{sec-evaluation}. Related work is surveyed in Section~\ref{sec-related}. We prove that integrating the VBR scheme maintains linearizability and lock-freedom in Appendix~\ref{sec-correctness} and~\ref{sec-checkpoint-extra}. Finally, an illustration of integrating VBR into a lock-free data-structure (Harris's linked-list~\cite{harris2001pragmatic}) appears in Appendix~\ref{sec-example}.


\section{Overview of VBR} \label{sec-overview}

The VBR memory reclamation scheme follows an optimistic approach where access to reclaimed objects is allowed. Optimistic approaches reduce the overhead but require care to guarantee correctness. 
First, VBR allows immediate reclamation of each retired object. There is no need to wait for guards to be lowered to make sure an object is reclaimable as in other methods. This property ensures that the amount of quarantined space waiting for reclamation is small and so lock-freedom can be obtained. Lock-freedom is also enabled because stalled threads do not delay reclamation of any object.
Second, on strongly-ordered systems (e.g., x86, SPARC TSO, etc.) VBR does not require a costly overhead on read or write accesses. No additional shared memory writes or memory synchronization fences are required with reads or writes to shared memory. This provides the high efficiency seen in the evaluation. However, on weakly-ordered systems (e.g., ARM, PowerPc, etc.), reads must be ordered using special CPU load or memory fence instructions~\cite{cppMem}.  
VBR requires a type-preserving allocator. I.e., a memory space allocated for a specific type is used only for the same type, even when re-allocated. 
The assumption of type preserving (see also~\cite{wen2018interval,cohen2015efficient,cohen2015automatic}) is necessary for applying our scheme, and is reasonable because data structure nodes are typically fixed-size nodes.
Retired nodes are not returned to the operating system. Instead, they are returned to a shared pool of nodes, from which they can be re-allocated by any thread.
As in~\cite{brown2015reclaiming,singh2021nbr}, a collection of local node pools (one per thread) is added to the shared pool. A thread accesses the shared pool only when it has no available nodes in its local pool.


Similarly to epoch-based reclamation, VBR maintains a global epoch counter, and as in~\cite{arbel2018harnessing,wen2018interval,ramalhete2017brief,nikolaev2020universal}, VBR tracks the birth epoch and retire epoch of each allocated node.
The birth epoch is determined upon allocation, and the retire epoch is set upon retirement.
The reclamation of a retired node does not involve any action. Upon an allocation of a node, a thread makes sure that the retire epoch of the node is smaller than the current global epoch. 
If it is not, then the thread increments the global epoch. This ensures that an object is allocated at a global epoch that is strictly larger than its previous retire epoch. Next, the thread 
re-allocates the object by updating its birth epoch with the current global epoch. 
This method guarantees that the ABA problem~\cite{michael2004aba} can only occur when the global epoch changes. Namely, when a thread encounters a  node during a data-structure traversal, it is guaranteed that this node has not been re-allocated during the traversal if the global epoch has not changed. 

VBR allows accessing reclaimed objects, while conservatively identifying reads that may access reclaimed nodes. To identify the access to a reclaimed node, each executing thread keeps track of the global epoch, by reading it upon most shared memory reads (as long as the epoch does not change, this read is likely to hit in the cache). When the thread observes an epoch change, it conservatively assumes that a value was read from a reclaimed memory and it applies a roll-back mechanism (described in Section~\ref{sec-modifications}), returning to a pre-defined checkpoint in its code.  
Since a node is always re-allocated at an epoch that is strictly larger than its former retirement, threads never rely on the content of stale values.

We now move on to handling optimistic writes. In addition to the birth epoch and retire epoch, each mutable field (e.g., node pointers) is associated with a version that resides next to it on the data structure node. During the execution, mutable fields are always updated atomically with their associated versions (using a wide CAS instruction). During allocation, after the birth epoch initialization, all versions of all mutable fields are set to the node's birth epoch. Throughout the life-cycle of a not-yet retired node, all of its versions remain greater than or equal to its birth epoch, and they never exceed its future retire epoch. Versions are decreased or increased during the execution in the following manner: when updating a pointer from a node $n$ to a node $m$, the pointer's version is set to the maximum birth epoch of the two nodes (either $n$'s or $m$'s). Notice that we assume that $n$'s pointer is never updated after its retirement (for more details, see Section~\ref{sec-assumptions}), and therefore, none of its pointers are assigned a version that exceeds their retirement epoch.

Let us consider the ABA problem for this versioning scheme. The concern is that re-allocations may result in an  erroneous success of CAS executions. For example, suppose in a linked list that a node $n$ points to another node, $m$, which in turn points to a third node, $k$. Now, suppose that a thread $T_1$ tries to remove $m$ by setting $n$'s pointer to point to $k$. Right before executing the removing CAS, $T_1$ is halted. While $T_1$ is idle, $T_2$ removes $m$ and then reallocates $m$'s space as a new node $d$. Next, $T_2$ inserts $d$ as a new node between $n$ and $k$. In the lack of versions, $T_1$'s CAS will now be erroneously successful. However, with versions it must fail. Since $d$'s birth epoch is necessarily bigger than $m$'s retire epoch, the version in the original pointer to $m$ must be smaller than the version assigned to the pointer when $d$ becomes its referent, 
and the CAS fails (for more details, see Appendix~\ref{sec-correctness}). We leave the rest of the discussion on handling unsuccessful CAS executions to Section~\ref{sec-algorithm}.
\section{Settings and Assumptions} \label{sec-settings}
In this section we describe our shared memory model and specify the assumptions a data structure must satisfy for integrating with our reclamation mechanism.

\subsection{System Model} \label{sec-model}
We use the basic asynchronous shared memory model, as described in~\cite{herlihy1991wait}. In this model, a fixed set of threads communicate through memory access operations.
Since threads may be arbitrarily delayed or may crash in the middle of their execution (which immediately halts their execution), they cannot assume anything about other threads' speed or whether or not other threads are active at all.
The shared memory is accessed via atomic instructions, provided by the hardware. Such instructions may be atomic reads and writes, the compare-and-swap (CAS) instruction and the wide-compare-and-swap (WCAS, which atomically updates two adjacent memory words, and is often supported in commodity hardware~\cite{yosifovich2017windows}) instruction.
The CAS operation receives three input arguments: an address of a certain word in memory, an expected value and a new value (both of the size of a single word). It then atomically compares the memory address content to the expected value, and if they are equal, it replaces it with the new received value. Otherwise, it does nothing. The WCAS operation operates in the same manner, on two adjacent memory words. 

Concurrent implementations provide different progress guarantees.
\textit{Lock-freedom} guarantees that as long as at least one thread executes its algorithm long enough, some thread will eventually make progress (e.g., complete an operation). This progress guarantee is not affected by the scheduler or even by the crash of all threads except for one. For a lock-free data structure to be truly lock-free,
it must rely on an allocation method that is also lock-free,
because otherwise a blocked allocation can prevent all threads from making progress.


A {\em data structure} represents a set of {\em items}, which are distinguished by unique keys, and are often arranged in some order. Each item is represented by a {\em node}, consisting of both mutable and immutable fields. In particular, each node has an immutable {\em key} field. The data-structure has a fixed set of {\em entry points} (e.g., the head of the linked-list in~\cite{harris2001pragmatic}), which are node pointers. 
A data structure provides the user with a set of operations for accessing it. Moreover, The user cannot access the data-structure in other ways, and the data structure operations never return a node reference.
An item that belongs to the data structure set of items must be represented by a node which is reachable from an entry point, by following a finite set of pointers. In particular, the data-structure nodes are accessible only via the entry points. However, a reachable node does not necessarily represent an item in the data structure set of items. We denote the removal of an item from the set of items that the data structure represent by {\em logical deletion}, and we denote the unlinking of a node from the data structure (i.e., making the node unreachable from the entry points) as {\em physical deletion}. 
E.g., in~\cite{harris2001pragmatic}, a special mechanism is used in order to mark reachable nodes as deleted. Once a node is marked, it stops representing an item in the data structure set of items (i.e., it is logically deleted), even though it is reachable from an entry point. 
A logically-deleted node is never logically re-inserted into the data-structure, even if it is physically re-inserted~\cite{sheffi2018scalable,timnat2012wait}. For example, if a node is logically deleted by marking one of its fields, 
then a marked node may be reinserted to the data structure, but its deleted mark cannot be removed. 

\subsection{Executions, Histories and Linearizability} \label{sec:linearizability}

A {\em step} can either be a shared-memory access by a thread (including the access input and output values), a local step that updates its own local variables, an invocation of an operation or the return from an operation (including the respective inputs and outputs).
We assume each step is atomic, so an {\em execution} $E=s_1 \cdot s_2 \cdot \ldots$
consists of a sequence of steps, assumed to start after an initial state, in which all data-structures are initialized and empty.
Given E, we further denote the finite sub-execution $s_1 \cdot s_2 \cdot \ldots \cdot s_i$ as $E_i$.

We follow~\cite{herlihy1990linearizability}, and model an execution E by its {\em history} H (and $E_i$ by $H_i$, respectively), which is the sub-sequence of operation invocation and response steps.
A history is {\em sequential} if it begins with an invocation step, and all invocations (except possibly the last one) have immediate matching responses. 
We assume that a concurrent system is associated with a {\em sequential specification}, which is a prefix-closed set of all of its possible sequential histories. A sequential history is {\em legal} iff it belongs to the sequential specification.
An invocation is {\em pending} in a given history if the history does not contain its matching response.
Given a history H, its sub-sequence excluding all pending invocations is denoted as complete(H).
An {\em extension} of a history H is a history constructed by appending responses to zero or more pending invocations in H. 
We further extend the notion of extensions, and say that an execution E' is an {\em extension} of an execution E if E is a prefix of E'.
In addition, given an execution E, EXT(E) is the set of all histories H' such that (1) H' is an extension of E's respective history, and (2) H' is the respective history of an extension of E.
Given a history H and a thread T, T's {\em sub-history}, denoted as H|T, is the sub-sequence of H consisting of all (and exactly) the steps executed by T.
Two histories H and H' are {\em equivalent} if for every thread T, H|T and H'|T are equal.
A history H is {\em well-formed} if for every executing thread T, H|T is a sequential history.
A well-formed history H is linearizable if it has an extension H' for which there exists a legal sequential history S such that (1) complete(H') is equivalent to S, and (2) if a response step precedes an invocation step in H, then it also precedes it in S.

\subsection{Implementation Assumptions} \label{sec-assumptions}

We focus on adding the VBR reclamation scheme to lock-free linearizable concurrent data-structure implementations. As in~\cite{wei2021constant}, we first assume that modifications are executed using the CAS instruction. No simple writes are used, and no other atomic instructions are supported. 
Consequently, our scheme does not support the use of other atomic primitives (such as $fetch\&add$ and $swap$).

\begin{assumption} \label{assumption-cas}
All updates occur only via CAS executions.
\end{assumption}




In general, as in~\cite{ramalhete2017brief}, we assume that all mutable fields of a removed node are invalidated, in order to prevent their future updates. It can be achieved either by marking them~\cite{harris2001pragmatic,michael2002high} or by self-linking (in the case of pointers). 
More formally:
\begin{assumption} \label{assumption-invalidate}
Node fields are invalidated (and become immutable) using a designated {\em invalidate()} method. This method receives as input a node field and invalidates it. The invalidation succeeds iff the field is valid and is not concurrently being updated by another thread. In order to check whether a certain field is invalid, a thread calls a designated {\em isValid()} method. Finally, given a node field, a thread separates the value from the (possible) invalidation mark by calling a designated {\em getField()} method.
\end{assumption}


Following the standard interface for manual reclamation~\cite{michael2004hazard,wen2018interval,cohen2015efficient,ramalhete2017brief,kang2020marriage}, applying our reclamation scheme to an existing implementation includes allocating nodes using an \textit{alloc} instruction and retiring nodes using a \textit{retire} instruction. Nodes are always retired before they can be reclaimed by the reclamation scheme. We assume that it is possible to retire each node only once. 
%
To sum up, in a similar way to~\cite{michael2004hazard}:
\begin{assumption} \label{assumption-life}
We assume the following life-cycle of a node $n$:
\begin{enumerate}
    \item \label{node-allocated} \textbf{Allocated}: $n$ is allocated by an executing thread, but is not yet reachable from the data-structure entry points. Once it is physically inserted into the data-structure, it becomes reachable.  
    \item \label{node-reachable} \textbf{Reachable} (optional): $n$ is reachable from the data structure entry points, but is not yet necessarily logically inserted into the data-structure (e.g.,~\cite{sheffi2018scalable,timnat2012wait}). When it is made logically included in the data structure it becomes Reachable and valid. 
    \item \label{node-valid} \textbf{Reachable and valid}: $n$ is reachable from the entry points and is considered logically in the data-structure (i.e., valid). At the end of this phase, fields of $n$ are invalidated. We think of $n$ as invalid when at least one of its mutable fields is invalid. 
    \item \label{node-invalid} \textbf{Invalid}: $n$ is logically deleted by a designated invalidation procedure, and all of its mutable fields are invalidated (e.g., by marking~\cite{harris2001pragmatic}). Once a field is invalidated, it becomes immutable. At the end of this phase, $n$ is unlinked (physically deleted) from the data structure. 
    \item \label{node-unlinked} \textbf{Unlinked}: At this point, $n$ is not reachable from the data-structure entry points, and therefore, it is not reachable from any other linked node. At the end of this phase it is retired by some thread. We assume a node is retired only once in an execution. After being retired the node is never linked back into the data structure. 
    \item \label{node-retired} \textbf{Retired}: $n$ has been retired, by a certain thread. We assume that only unlinked nodes can be retired. 
\end{enumerate}

\end{assumption}
%

Notice that, as discussed in~\cite{fraser2004practical,cohen2015automatic}, a node can be physically removed and re-inserted into the data-structure several times during stage~\ref{node-invalid}. However, a {\em retire} instruction is issued on a node only after it is physically removed for the last time and it will not be re-inserted thereafter. 

Finally, we assume that a thread does not use data on nodes without occasionally checking that the nodes are valid. For our scheme to work, we require this check after modifying the data structure. We assume that if a thread performs a successful modification of the data structure, and if it has some locally saved pointers that were read prior to the modification, then the thread makes limited use of these pointers. Actually, we do not even need to impose the restriction on all modifications. Restrictions are needed only for "important" modifications that cannot be rolled back. Such modifications are called {\em rollback-unsafe} and they are formally defined in Section~\ref{sec-checkpoints} below.   
In particular: 
\begin{assumption} \label{assumption-mutable}
If thread $T$ executes any rollback-unsafe modification after updating a local pointer $p$, then a future use of $p$ is limited. Suppose $p$ references a node $n$, then a future (i.e., after the rollback-unsafe modification) read of one of $n$'s mutable fields by $T$ is allowed only if the read is followed by an {\em isValid()} call, and if it returns FALSE, the field content is not used by $T$. 
\end{assumption}


While "not using" the content of a read field is intuitively clear, let us also formally say that the content of a read field is not used by a thread $T$, if $T$'s behavior is indistinguishable from its behavior when reading the $\bot$ sign instead of the actual value read.

Note that even after a rollback-unsafe update, $T$ is allowed to use the content of fields that were read before the modification. However, after the modification, $T$ is not allowed to dereference a local pointer and read values from the referenced node without checking the validity of the node. 
For example, $T$ is allowed to use previously read pointers as expected values of a CAS, or as the target of a write operation. $T$ can traverse a list in a wait free manner (since there is no modification involved). $T$ can trim all invalid nodes along a traversal. This is allowed since such modifications are  rollback-safe, and even if they were not, trimming includes checking the validity of the traversed nodes. 
Finally, T can search for the set of expected predecessors and successors in all of the skip-list upper levels, before it starts the physical insertion of a new node~\cite{linden2013skiplist,herlihy2020art,fraser2004practical}. 
All known lock-free data structures that we are aware of (e.g., ~\cite{harris2001pragmatic,natarajan2014fast,shalev2006split,ellen2010non,michael2002high,linden2013skiplist,michael1996simple,fraser2004practical,herlihy2020art}) satisfy Assumption~\ref{assumption-mutable}. 

\section{VBR: Version Based Reclamation} \label{sec-algorithm}

In this section we present the VBR mechanism: a lock-free recycling support for lock-free linearizable~\cite{herlihy1990linearizability} data-structures. 
We start by describing the reclamation scheme and the modifications applied to the nodes' representation in Section~\ref{sec-allocator}, and continue with the modifications applied to the data-structure operations in Section~\ref{sec-modifications}. 
In Section~\ref{sec-checkpoints} we define the notion of code checkpoints and show how to insert them into an existing linearizable implementation. In Section~\ref{sec-read} we go over the necessary adjustments to read operations (from shared variables). Handling update operations is described in Section~\ref{sec-write}.
A full example API appears in Figure~\ref{fig:interface}.
For ease of presentation, we refer to data-structures for which each node has a single immutable field (the node's key) and a single mutable field (the node's next pointer), and nodes' invalidation is executed via the marking mechanism~\cite{harris2001pragmatic}. However, this interface can be easily extended to handle multiple immutable and mutable fields, and other invalidation schemes. 
We present a full correctness proof for Theorem~\ref{theorem-maintain} in Appendix~\ref{sec-correctness}.

\begin{theorem} \label{theorem-maintain}
Given a lock-free linearizable data-structure implementation, that satisfies all of the assumptions presented in Section~\ref{sec-assumptions}, the implementation remains lock-free and linearizable after integrating it with VBR according to the modifications described in Sections~\ref{sec-allocator}-\ref{sec-modifications}.
\end{theorem}

\subsection{The Reclamation Mechanism} \label{sec-allocator}

VBR uses a shared epoch counter, denoted $e$, incremented periodically by the executing threads. In addition, each executing thread keeps track of the global epoch using a local {\em my\_e} variable.
Each node is associated with {\em birth\_epoch} and {\em retire\_epoch} fields. Its birth epoch contains the epoch seen by the allocating thread upon its allocation, and its retire epoch contains the epoch seen by the thread which removed this node from the data structure, right before its retirement.

\ignore{
\begin{figure}[!ht]
	\begin{algorithmic}[1]\footnotesize
		\State \textbf{Class} Node \label{class-node}
		\Indent
            \State long birth\_epoch
            \State long retire\_epoch
            \State int key
            \State struct next
            \Indent
                \State Node* data
                \State long version
            \EndIndent
        \EndIndent
	\end{algorithmic}
	\caption{\small An example node class}
	\label{fig:node}
\end{figure}
}

Similarly to IBR's  \textit{Tagged Pointer IBR}~\cite{wen2018interval}, we add a version field adjacent to each mutable field (e.g., node pointers).
The field's version is guaranteed to always be equal to or greater than the node's birth epoch, and equal to or smaller than its eventual retire epoch (if there exists any). 
The field's data and its associated version are always updated together (using WCAS). E.g., see lines~\ref{alloc-cas},~\ref{cas-wcas} in Figure~\ref{fig:interface}.

Handling reclamation at the operating system level often requires using locks (unless it is configured to ignore certain traps). Therefore, for maintaining lock-freedom, VBR uses a user-level allocator. Retired nodes are inserted into manually-managed node pools~\cite{treiber1986systems,herlihy1993methodology} for future re-allocation.
We use a type-preserving allocator. I.e., a memory space allocated for a specific type is used only for the same type, even when re-allocated. 

Each thread's allocation and reclamation mechanism works as follows.
Besides sharing a global nodes pool~\cite{treiber1986systems,herlihy1993methodology}, each executing thread maintains a local pool of retired nodes, from which it retrieves reclaimed nodes for re-allocations. When the local pool becomes large enough, retired nodes may be moved to a global pool of retired nodes, allowing re-distribution of reclaimed nodes between the threads. 
When retiring a node, it is possible to re-allocate this node immediately. However, we use a local retired list to stall its re-allocation for a while. This allows infrequent increments to the global epoch counter, which improves performance. A retired node is therefore  
added to the private list of retired nodes. When the size of the retired list exceeds a pre-defined threshold, it is appended as a whole to the thread's local allocation pool, becoming available for allocation. 

The full allocation method appears in lines~\ref{alloc-call}-\ref{alloc-return} of Figure~\ref{fig:interface}. First, the thread reads the retire epoch of the next available node in its allocation pool. If it is equal to the shared epoch counter, then the thread increments the shared epoch counter using CAS (line~\ref{alloc-update-e}) and executes a rollback to the previous checkpoint (for more details, see Section~\ref{sec-checkpoints}). This makes sure that the birth epoch of a new node is larger than the retire epoch of the node that was previously allocated on the same memory space. 
If the CAS is unsuccessful, then another thread has incremented the global epoch value and there is no need to try incrementing it again. 
If $e$ is bigger than the retired node's retire epoch, the thread sets the new node's birth epoch to its current value. After setting the node's birth epoch, its next pointer version is set to this value, along with an initialization of its data to NULL in line~\ref{alloc-cas}. 
Due to Assumption~\ref{assumption-life} (the mutable fields of a node become immutable before it is retired), the WCAS executed in line~\ref{alloc-cas} is always successful. Finally, the key field is set to the key value received as input. 

The {\em retire} method appears in lines~\ref{retire-call}-\ref{retire-return-cp}. First, the retiring thread makes sure that the node is not already retired in line~\ref{retire-if-retired} (for more details, see Appendix~\ref{sec-correctness}). Then, it sets the node's retire epoch to be the current global epoch, and appends the retired node to its local retired nodes list. In case its local copy of the global epoch counter is not up to date, it performs a checkpoint rollback in line~\ref{retire-return-cp} (for more details, see Section~\ref{sec-checkpoints}).

\begin{figure}[!ht]
	\begin{algorithmic}[1]\footnotesize
		\State \textbf{alloc}(int key) \label{alloc-call}
		\Indent
		    \State n := alloc\_list $\rightarrow$ next \label{alloc-retrieve}
		    \State if (n $\rightarrow$ retire\_epoch $\geq$ my\_e) \label{alloc-if}
		    \Indent
		        \State CAS(\&e, my\_e, my\_e + 1) \label{alloc-update-e}
		        \State alloc\_list $\rightarrow$ next := n \label{alloc-realloc}
		        \State return to checkpoint \label{alloc-return-cp} \Comment{Checkpoint rollback}
		    \EndIndent
		    \State n $\rightarrow$ birth\_epoch := my\_e \label{alloc-update-birth}
		    \State n $\rightarrow$ retire\_epoch := $\bot$ \label{alloc-update-retire}
		        \State WCAS(\&(n $\rightarrow$ next), $\langle$n $\rightarrow$ next.data, n $\rightarrow$ next.version$\rangle$, $\langle$NULL, my\_e$\rangle$) \label{alloc-cas}
            \State n $\rightarrow$ key := key \label{alloc-key}
            \State return n \label{alloc-return}
        \EndIndent
        \Statex
        
        \State \textbf{retire}(Node* n, long n\_b) \label{retire-call}
        \Indent
            \State if (n $\rightarrow$ birth\_epoch > n\_b || n $\rightarrow$ retire\_epoch $\neq \bot$) return \label{retire-if-retired} \Comment{Avoiding double retirements}
            \State n $\rightarrow$ retire\_epoch := e.get() \label{retire-update}
            \State retired\_list $\rightarrow$ next := n \label{retire-retire}
            \State if (n $\rightarrow$ retire\_epoch > my\_e) return to checkpoint \label{retire-return-cp} \Comment{Checkpoint rollback}
        \EndIndent
        \Statex

        \State \textbf{getNext}(Node* n) \label{get-ref-call} 
        \Indent
            \State n\_next := unmark(n $\rightarrow$ next.data) \label{get-ref-get-ref}
            \State n\_next\_b := n\_next $\rightarrow$ birth\_epoch \label{get_next_birth}
		    \State if (my\_e $\neq$ e.get()) return to checkpoint \label{read-return-cp} \label{read-if} \Comment{Checkpoint rollback}
            \State return n\_next, n\_next\_b \label{get-ref-return}
        \EndIndent
        \Statex
        
        \State \textbf{getKey}(Node* n) \label{get-key-call} 
        \Indent
            \State n\_key := n $\rightarrow$ key \label{get-key-get-key}
            \State if (my\_e $\neq$ e.get()) \label{get-key-if} return to checkpoint \label{get-key-return-cp} \Comment{Checkpoint rollback}
            \State return n\_key \label{get-key-return}
        \EndIndent
        \Statex
        
        \State \textbf{isMarked}(Node* n, long n\_b) \label{is-valid-call}
        \Indent
            \State res := isMarked(n $\rightarrow$ next.data) \label{is-valid-is-valid}
            \State if (n $\rightarrow$ birth\_epoch $\neq$ n\_b) return TRUE \label{is-valid-if} \label{is-valid-return-false} \Comment{The node is already removed}
            \State return res \label{is-valid-return-res}
        \EndIndent
        \Statex
        
		\State \textbf{update}(Node* n, long n\_b, Node* exp, long exp\_b, Node* new, long new\_b) \label{cas-call}
		\Indent
		    \State exp\_v := max \{ n\_b,  exp\_b \} \label{cas-max-exp}
		    \State new\_v := max \{ n\_b,  new\_b \} \label{cas-max-birth}
            \State return WCAS(\&(n $\rightarrow$ next), $\langle$ exp, exp\_v $\rangle$, $\langle$ new, new\_v $\rangle$) \label{cas-wcas}
        \EndIndent
        \Statex
        
        \State \textbf{mark}(Node* n, long n\_b) \label{invalidate-call}
		\Indent
		    \State exp := unmark(n $\rightarrow$ next.data) \label{mark-read-next}
		    \State exp\_v := max \{ n\_b,  exp $\rightarrow$ birth\_epoch \} \label{mark-max-exp}
		    \State if (n $\rightarrow$ birth\_epoch $\neq$ n\_b) return FALSE \label{mark-return-false} \label{invalidate-if-marked} \Comment{The node is already removed}
		    \State new := mark(exp) \label{mark}
            \State return WCAS(\&(n $\rightarrow$ next), $\langle$ exp, exp\_v\ $\rangle$, $\langle$ new, exp\_v $\rangle$) \label{invalidate-wcas}
        \EndIndent
	\end{algorithmic}
	\caption{\small An example VBR interface}
	\label{fig:interface}
\end{figure}

\subsection{Code Modifications} \label{sec-modifications}

Unlike former reclamation methods, VBR allows both optimistic reads and optimistic writes. Namely, the executing threads may sometimes access a previously reclaimed node, and either read its stale values or try to update it.
To the best of our knowledge, there exists no other scheme which allows optimistic writes, and optimistic reads are allowed only in~\cite{cohen2015efficient,cohen2015automatic,Cohen18}. It is therefore the first algorithm that allows full speculative execution in a memory manager, for the purpose of improved performance. 

Accesses to reclaimed nodes should be monitored and handled. Our versioning mechanism ensures that a write to a previously reclaimed node always fails, and that stale values that are read from reclaimed nodes are always ignored. 
This gives rise to an additional problem -- when failing to write a value or to read a fresh value due to an access to a reclaimed node, the program needs to move control to an adequate location. 
This problem does not arise with epoch based reclamation because a thread never fails due to a test that the memory reclamation scheme imposes.
Failures that arise due to optimistic access are not part of the original lock-free concurrent data structure. Interestingly, deciding how to treat failed reads or writes is very easy in practice. We could easily modify lock-free data structures that satisfy the assumptions presented in Section~\ref{sec-assumptions} (e.g., ~\cite{harris2001pragmatic,natarajan2014fast,shalev2006split,ellen2010non,michael2002high,linden2013skiplist,michael1996simple,fraser2004practical}), at a minimal performance cost.
However, while presenting this scheme, we would also like to propose a general manner to handle failed accesses. We are going to define the notion of execution checkpoints. 
Upon accessing an allegedly stale value, the code just rolls back to the appropriate checkpoint.
This problem is given general treatment in the format of a normalized form assumption in~\cite{cohen2015efficient,cohen2015automatic}, and in a total separation between read and write phases during the execution in~\cite{singh2021nbr}. Although the VBR scheme can be applied to implementations that adhere to both models, both of them require extensive modifications to the original program's structure. Therefore, we propose a new method, which is more general and makes less assumptions on the given implementation. Our method is to carefully define program checkpoints.

\subsubsection{Defining Checkpoints} \label{sec-checkpoints}



As briefly described in Section~\ref{sec-overview}, and as will be formally defined below, VBR occasionally requires a rollback to a predefined checkpoint. In order to install checkpoints in a given code in an efficient manner, one needs to be able to distinguish important shared-memory accesses that cannot be rolled back from non-important accesses that allow rolling back (a simple decision that all modifications are not safe to rollback is correct, but may reduce performance). 
The notion of important shared-memory accesses is similar to the definition of an {\em owner CAS} in~\cite{timnat2014practical}, and shares some mutual concepts with the {\em capsules} definition, given in~\cite{ben2019delay,blelloch2018parallel}.
Informally, non-important shared-memory accesses (that can be rollbacked) either do not affect the shared memory view (e.g., reading from the shared-memory) or do not have any meaningful impact on the execution flow. 
For example, consider Hariss's implementation of a linked-list~\cite{harris2001pragmatic}.
The physical removal of a node, i.e., trimming the node from the list after it has been marked, can be safely rollbacked. If we try the same trim again, it will simply fail, and if we rollback further, this trim will not even be attempted. 
However, the (successful) insertion of a new node into the list and the (successful) marking of a node for logical deletion are both important and are not rollback-safe. In both cases, performing a rollback right after the successful update would result in a non-linearizable history (as the inserter or remover would not return TRUE after successfully inserting or removing the node, respectively).
We now define the notion of {\em rollback-safe steps} in a given execution $E$ with a respective history $H$. 

\begin{definition}[Rollback-Safe Steps] \label{definition-rollback-safe}
We say that $s_i$ is a rollback-safe step in an execution E if EXT($E_i$)=EXT($E_{i-1}$).
\end{definition}


If $s_i$ is not a rollback-safe step, then we say that it is a {\em rollback-unsafe step}.
According to Definition~\ref{definition-rollback-safe}, if $s_i$ is a rollback-safe step, executed by a thread T during $E$, then T can safely perform a {\em local rollback step} after $s_i$. I.e., right after executing $s_i$ by T, T can restore the contents of all of its local variables and program counter (assuming they were saved before $s_i$), and the obtained execution would have the same set of corresponding history extensions. 
Note that, by Definition~\ref{definition-rollback-safe}, local steps, shared memory reads and unsuccessful memory updates (CAS executions returning FALSE) are considered as rollback-safe steps.

We extend Definition~\ref{definition-rollback-safe} in the following manner: a thread T can rollback to any previously saved set of local variables (including its program counter), as long as it has not performed any rollback-unsafe steps since they had been saved. I.e., as we prove in Appendix~\ref{sec-correctness}, it can safely rollback to its last visited checkpoint. 


Given a code for a concurrent data-structure, checkpoints are first installed after some shared memory update instructions, in the following manner: let $l$ be a shared-memory update instruction (i.e., a CAS instruction). If there exists an execution $E=s_1 \cdot \ldots$ such that $l$ is executed successfully during a step $s_i$ (i.e., the CAS execution returns TRUE), and $s_i$ is a rollback-unsafe step in $E$, then a checkpoint is installed right after $l$.
The installation of a checkpoint includes a check that the update is indeed successful (the CAS execution returns TRUE). If it is, then the checkpoint reference is updated (to the current value of the program counter), and all  local variables are saved for a future restoration\footnote{It is unnecessary to save uninitialized variables and variables that are not used anymore}. If the update is not successful (the CAS execution returns FALSE), then nothing is done. 
Checkpoints are also installed in the beginning of each data-structure operation. As opposed to the first type of checkpoint triggers, the installation does not depend on anything when done upon an operation invocation. Recall that, by Definition~\ref{definition-rollback-safe}, an operation invocation is always a rollback-unsafe step. 
After rolling back to a checkpoint, the thread updates its local copy of the global epoch, recovers its set of local variables, and continues its execution\footnote{Right before a thread rolls-back to its previous checkpoint,
it handles some unlinked nodes for guaranteeing VBR's robustness. As this issue does not affect correctness, we move this discussion to Appendix~\ref{sec-checkpoint-extra}.}. 

\subsubsection{Read Methods} \label{sec-read}
As threads may access stale values, the only way to avoid relying on a stale value is to constantly check that the node from which the value was read has not been re-allocated. 
In a conservative way, we think of a read instruction as potentially reading a stale value if the global epoch number changed since the last checkpoint time. A read of a stale value must imply a change of the global epoch number because the birth epoch of a node is strictly larger than the retirement epoch of a previous node that resides on the same memory space. 
Therefore, the shared epoch counter $e$ is read upon each operation invocation (see Section~\ref{sec-checkpoints}), each node retirement, and after certain allocations and reads from the shared memory. 
Since a node cannot be allocated during an epoch in which a node, previously allocated from the same memory address, is not yet removed from the data-structure (see lines~\ref{alloc-if}-\ref{alloc-return-cp} in Figure~\ref{fig:interface}), as long as the global epoch, read before the read of a node, is equal to the one read after the read of the node, it is guaranteed that the node's value is not stale. 
In general, when reading a node pointer into a local variable, it is always saved together with the node's birth epoch, as the node is represented by its birth epoch as well. 

W.l.o.g. and for simplifying our presentation, we assume each node originally consists of an immutable key field and a mutable next pointer field, and that a node is invalidated using the {\em mark()} method~\cite{harris2001pragmatic}. Therefore, there are roughly three types of read-only accesses in the original reclamation-free algorithm. The first type is the read of a node via the next pointer of its predecessor, the second one is the read of a node's key, and the third one is the read of a node's mark bit. We install the {\em getNext()} method instead of each pointer read in the original code, the {\em getKey()} method instead of each key read, and a new {\em isMarked()} method instead of the original one. Accesses to other (mutable or immutable) node fields should be very similar, and therefore require the same treatment.

The code that should be installed instead of each pointer read in the original code, appears in lines~\ref{get-ref-call}-\ref{get-ref-return}, and the code that should be installed instead of each key read, appears in lines~\ref{get-key-call}-\ref{get-key-return}. Both methods receive a pointer to the target node (assumed to be given as an unmarked pointer). First, the next node and its birth epoch (or the key, respectively) are saved in local variables. Then, the global epoch is read and compared to the previous recorded epoch. If the epoch has changed since the previous read, then the values may be stale, and the execution returns to the last checkpoint. Otherwise, the data is returned in line~\ref{get-ref-return} (or line~\ref{get-key-return}, respectively).

The {\em getKey()} method gives rise to another issue.
In a similar way to Assumption~\ref{assumption-mutable}, if a thread $T$ installs a pointer to a node $n$, then $T$ cannot read $n$'s immutable fields after executing a rollback-unsafe step. However, when dealing with immutable fields, as opposed to Assumption~\ref{assumption-mutable}, we do not assume that the original reclamation-free implementation satisfies this assumption. Instead, when $T$ installs a local pointer to $n$, it immediately calls the {\em getKey()} method and saves its output along with its local pointer to $n$ (unless $n$'s key has already been read into one of $T$'s local variables). Future reads of $n$'s key are replaced with reads of the saved value.

The code that should be installed instead of each {\em isMarked()} call in the original code, appears in lines~\ref{is-valid-call}-\ref{is-valid-return-res}.
It also receives an unmarked pointer to the target node, and additionally, its birth epoch. It first checks whether the node's next pointer is indeed marked (line~\ref{is-valid-is-valid}), using the original {\em isMarked()} method, which receives the actual allegedly marked pointer and checks if it is marked. Then the node's birth epoch is read, for guaranteeing that the given node is the correct one (and not another one, allocated from the same memory space). If it is not, then the target node has certainly been marked in the past, and the method returns TRUE in line~\ref{is-valid-return-false}. Otherwise, it returns the answer received in line~\ref{is-valid-is-valid}.
As this method returns a correct answer regardless of epoch changes or the retirement of the target node, it does not read the global epoch nor returns to the last checkpoint.

\subsubsection{Update Methods} \label{sec-write}

By Assumption~\ref{assumption-cas}, all data structure updates are executed using CAS instructions. We consider two types of pointer updates. The first type, depicted in lines~\ref{cas-call}-\ref{cas-wcas} of Figure~\ref{fig:interface}, is the update of an unmarked pointer. The second type (lines~\ref{invalidate-call}-\ref{invalidate-wcas}) is the marking of an unmarked pointer. 
The update of other mutable fields can be similarly implemented. In particular, the version of non-pointer mutable fields should always be equal to the node's birth epoch (which makes such fields much easier to handle).

The {\em update()} method replaces the original pointer update via a single CAS instruction. 
It receives pointers to the target node, its expected successor and the new successor, together with their respective birth epochs. All three pointers are assumed to be unmarked. The expected and new pointer versions are calculated in the same manner (lines~\ref{cas-max-exp}-\ref{cas-max-birth}): the maximum birth epoch of the target node and successor node. The next field is either successfully updated or remains unchanged in line~\ref{cas-wcas}. 
In Appendix~\ref{sec-correctness} we prove that the target node's next pointer is updated iff (1) it has not been reclaimed yet, (2) it is not marked, and (3) it indeed points to the expected node (including the given birth epoch).

The {\em mark()} method marks an unmarked {\em next} pointer, without changing its pointed node. It receives the target node and its birth epoch. 
The actual marking is executed in line~\ref{invalidate-wcas}. It uses the unmarked and marked variants of the pointer, and does not change the pointer's version (calculated in line~\ref{mark-max-exp}). In Appendix~\ref{sec-correctness} we prove that the target node is marked iff (1) it has not been reclaimed yet, (2) it is not marked, and (3) it indeed points to the expected node (read in line~\ref{mark-read-next}), just before the marking.

\begin{figure*}[tbh]

      \begin{subfigure}{0.325\textwidth}
         \includegraphics[width=\textwidth]{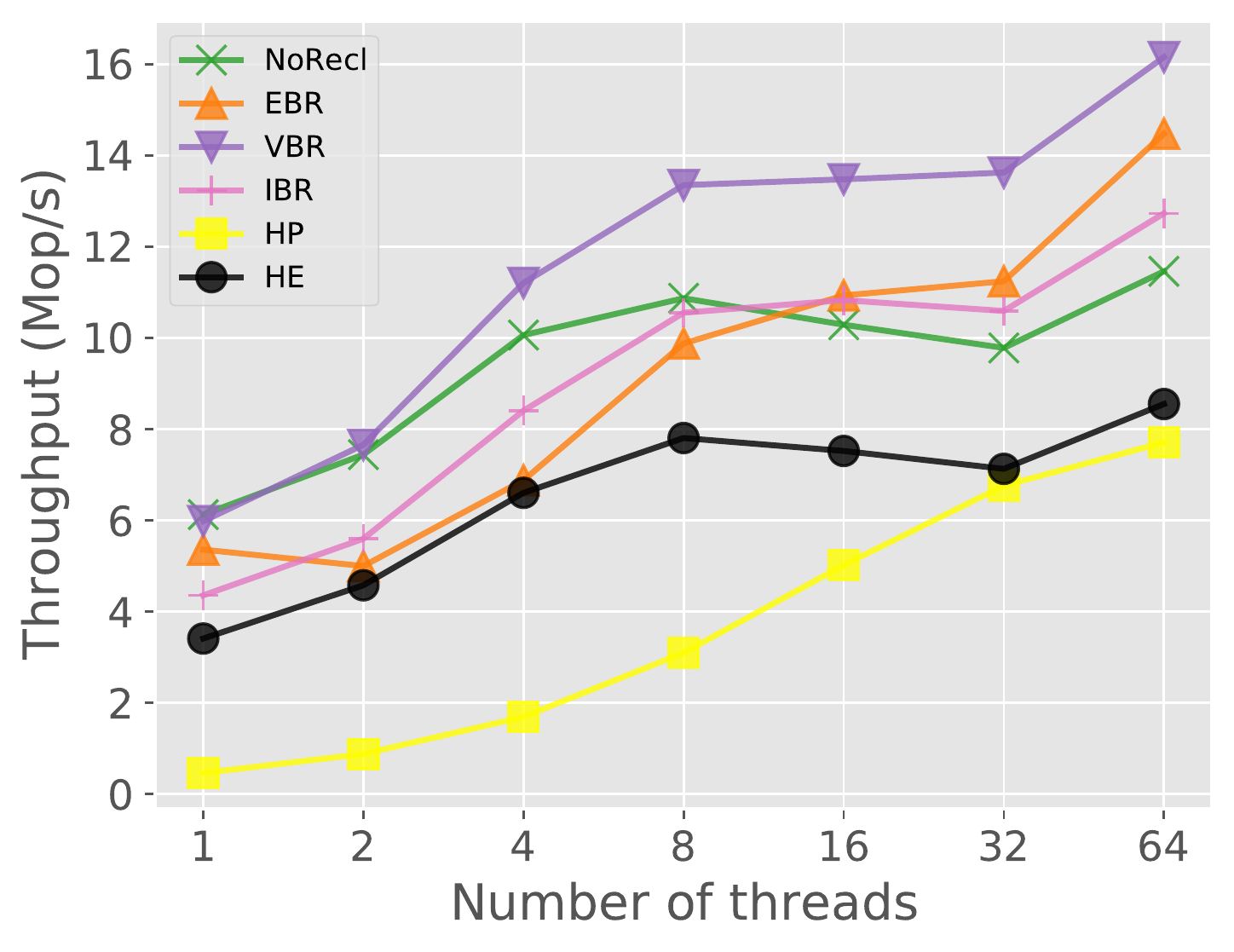}
        \caption{Linked-list. Key range: 256. 10\% inserts 10\% deletes 80\% reads.}
        \label{fig:list-256-20}
      \end{subfigure}
     \hfill    
	\begin{subfigure}{0.325\textwidth}
    	\includegraphics[width=\textwidth]{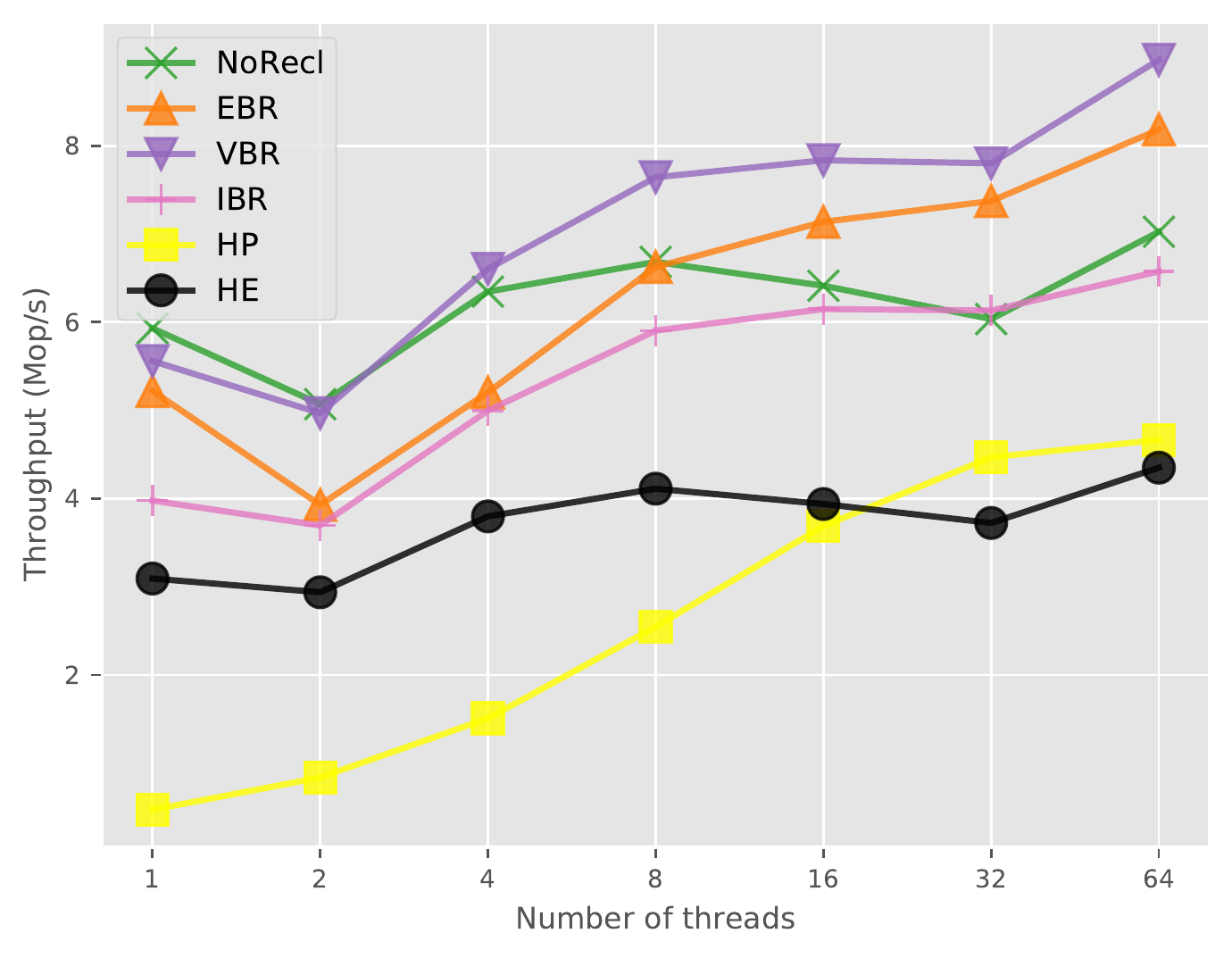}
	\caption{Linked-list. Key range: 256. 25\% inserts 25\% deletes 50\% reads.}
	\label{fig:list-256-50}
    \end{subfigure} 
 \hfill
     \begin{subfigure}{0.325\textwidth}
    \includegraphics[width=\textwidth]{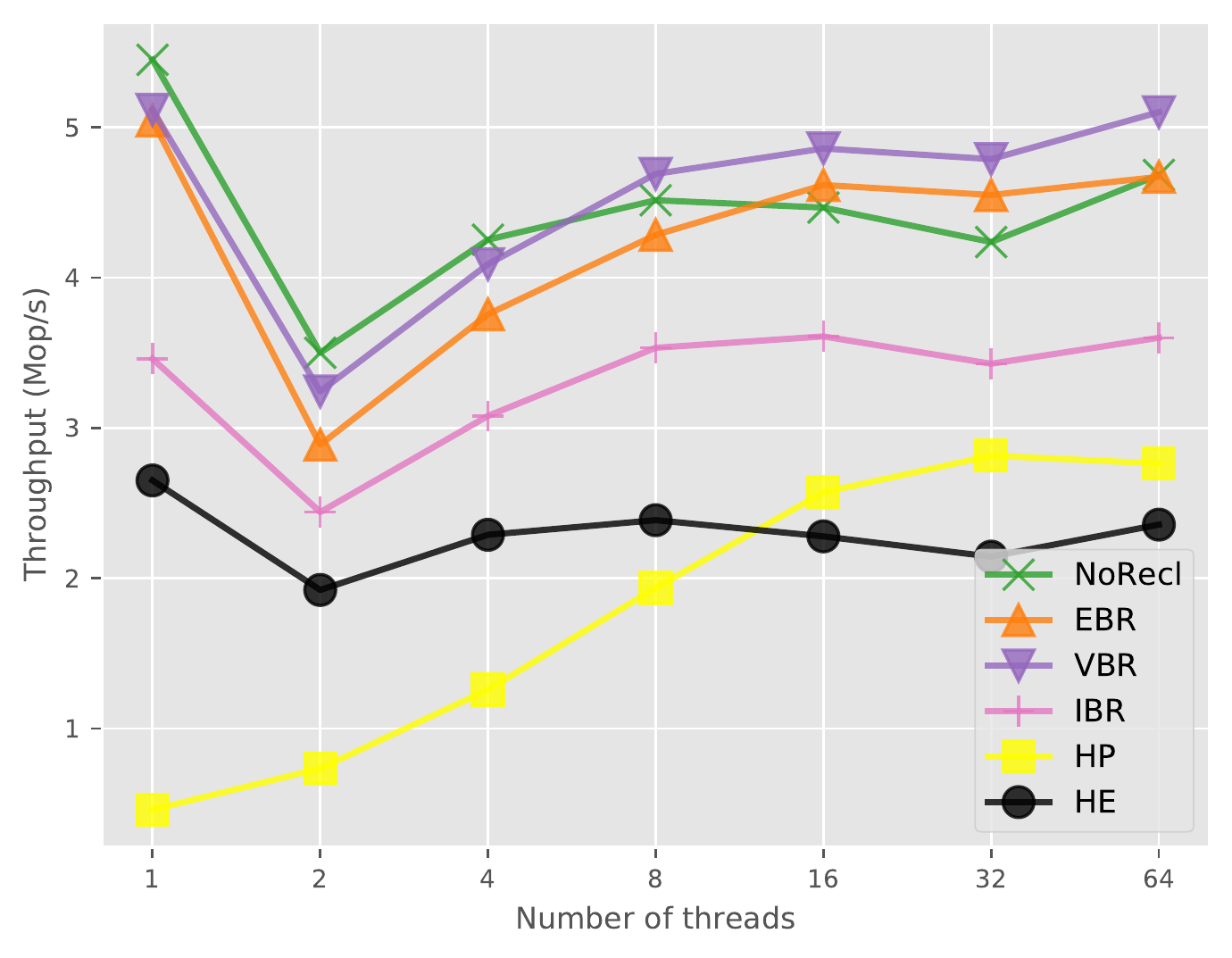}
	\caption{Linked-list. Key range: 256. 50\% inserts 50\% deletes.}
	        \label{fig:list-256-100}
    \end{subfigure}

     \begin{subfigure}{0.325\textwidth}
         \includegraphics[width=\textwidth]{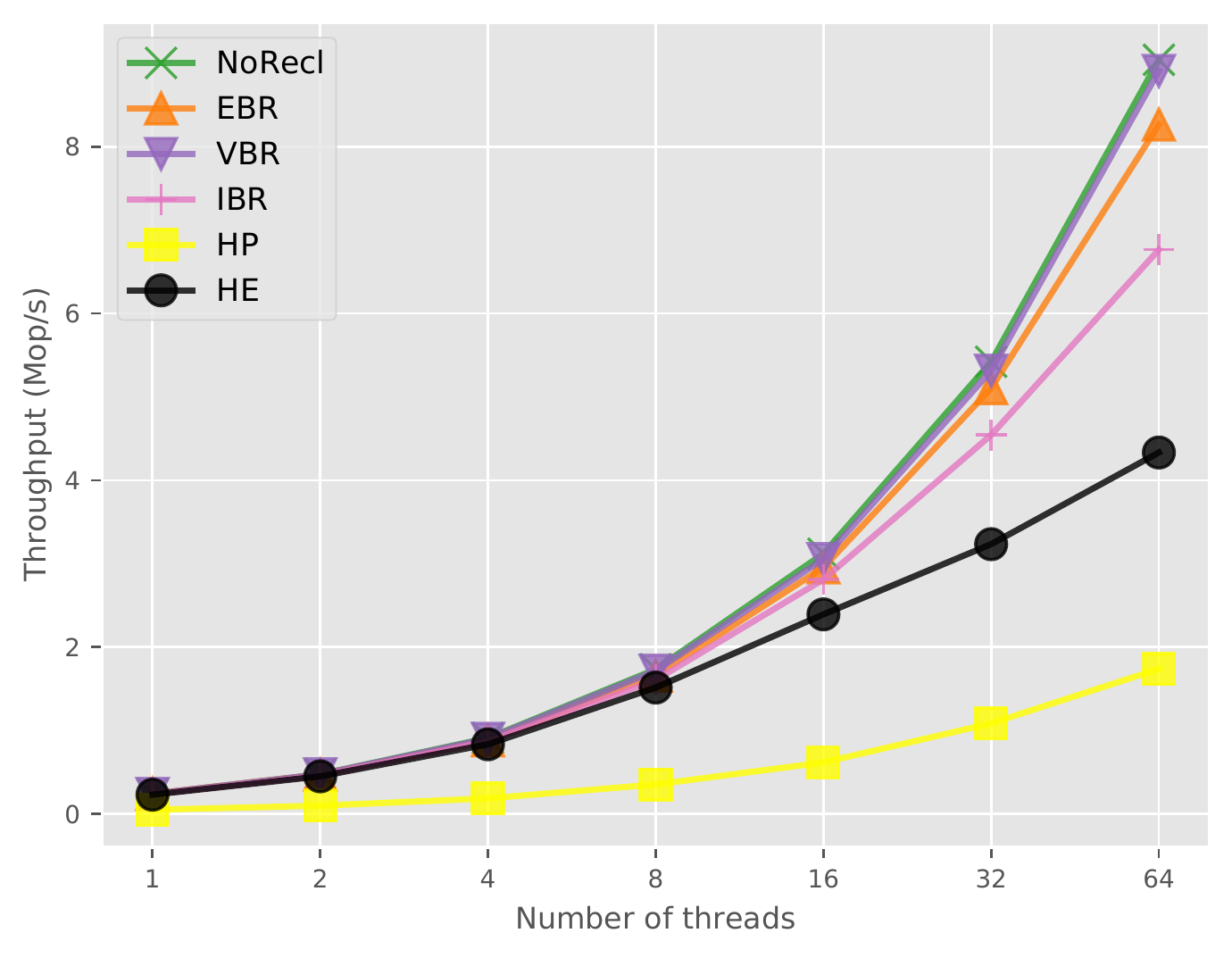}
        \caption{Skiplist. Key range: 10K. 10\% inserts 10\% deletes 80\% reads.}
        \label{fig:skiplist-10k-20}
      \end{subfigure}  
	\hfill
     	\begin{subfigure}{0.325\textwidth}
  	\includegraphics[width=\textwidth]{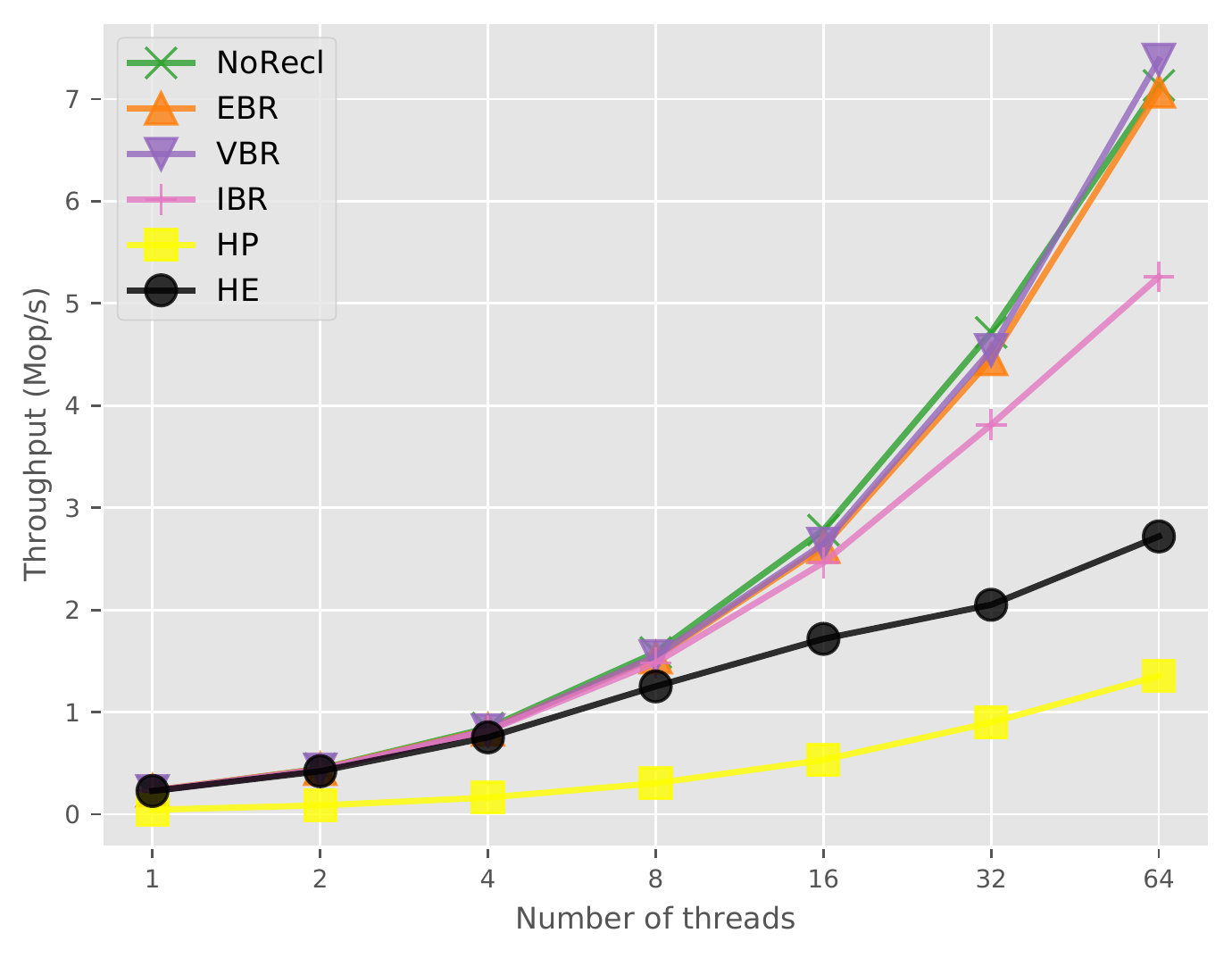}
  	\caption{Skiplist. Key range: 10K. 25\% inserts 25\% deletes 50\% reads.}
	\label{fig:skiplist-10k-50}
    \end{subfigure}
	\hfill    
      \begin{subfigure}{0.325\textwidth}
         \includegraphics[width=\textwidth]{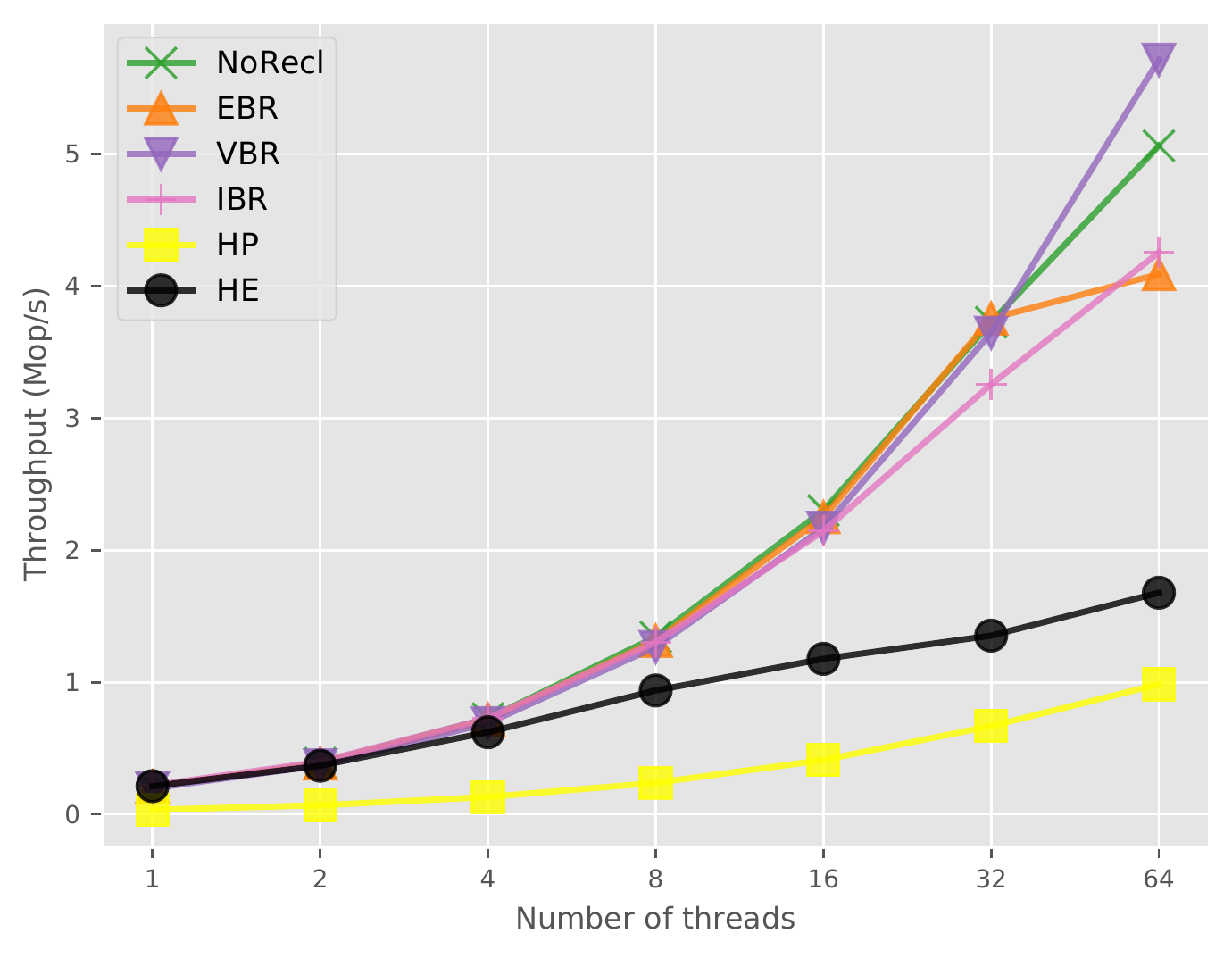}
        \caption{Skiplist. Key range: 10K. 50\% inserts 50\% deletes.}
        \label{fig:skiplist-10k-100}
      \end{subfigure}

      \begin{subfigure}{0.325\textwidth}
      	\includegraphics[width=\textwidth]{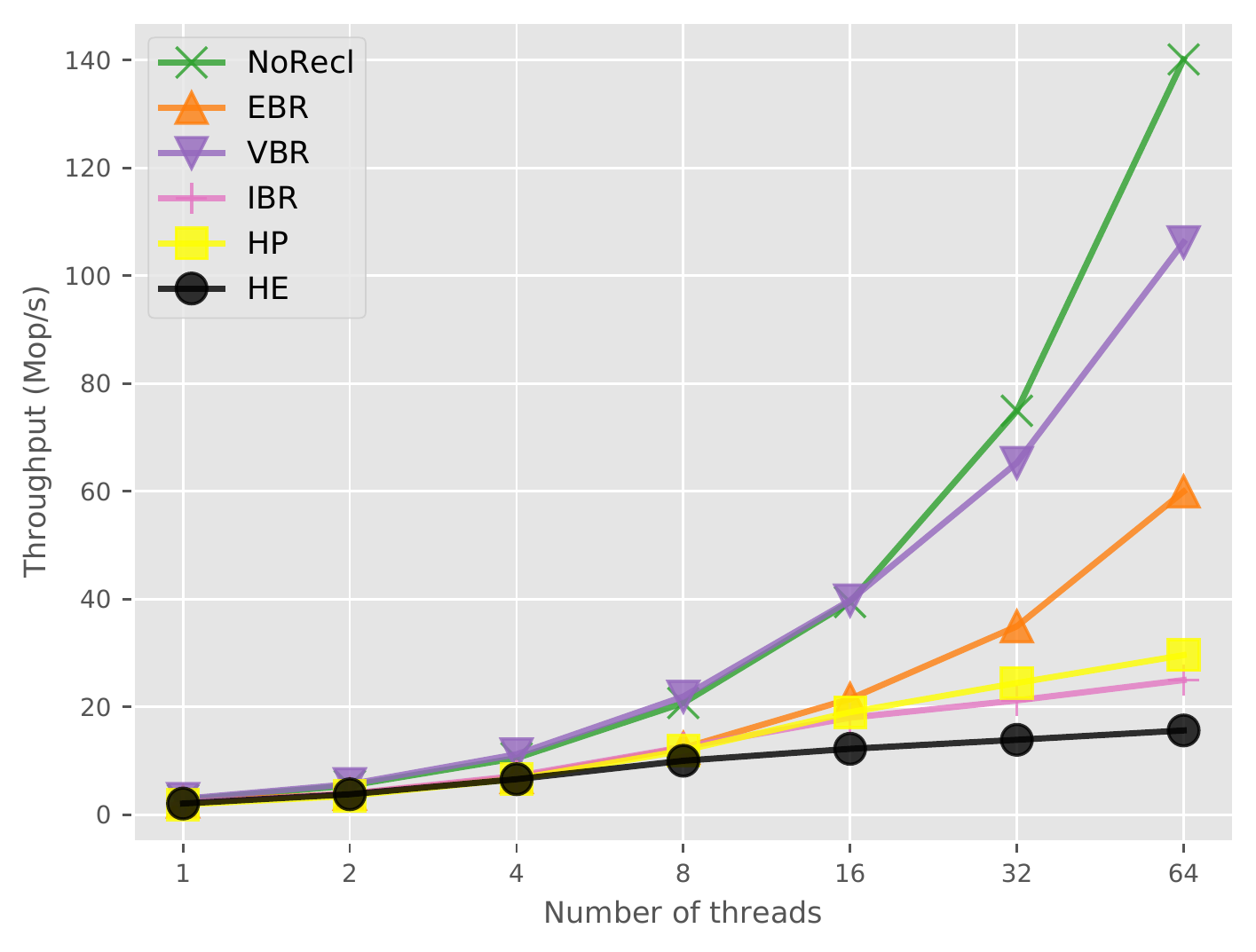}
        \caption{Hash table. Key range: 10M. 10\% inserts 10\% deletes 80\% reads.}
	\label{fig:hashmap-10m-20}
      \end{subfigure}  
\hfill
      \begin{subfigure}{0.325\textwidth}
         \includegraphics[width=\textwidth]{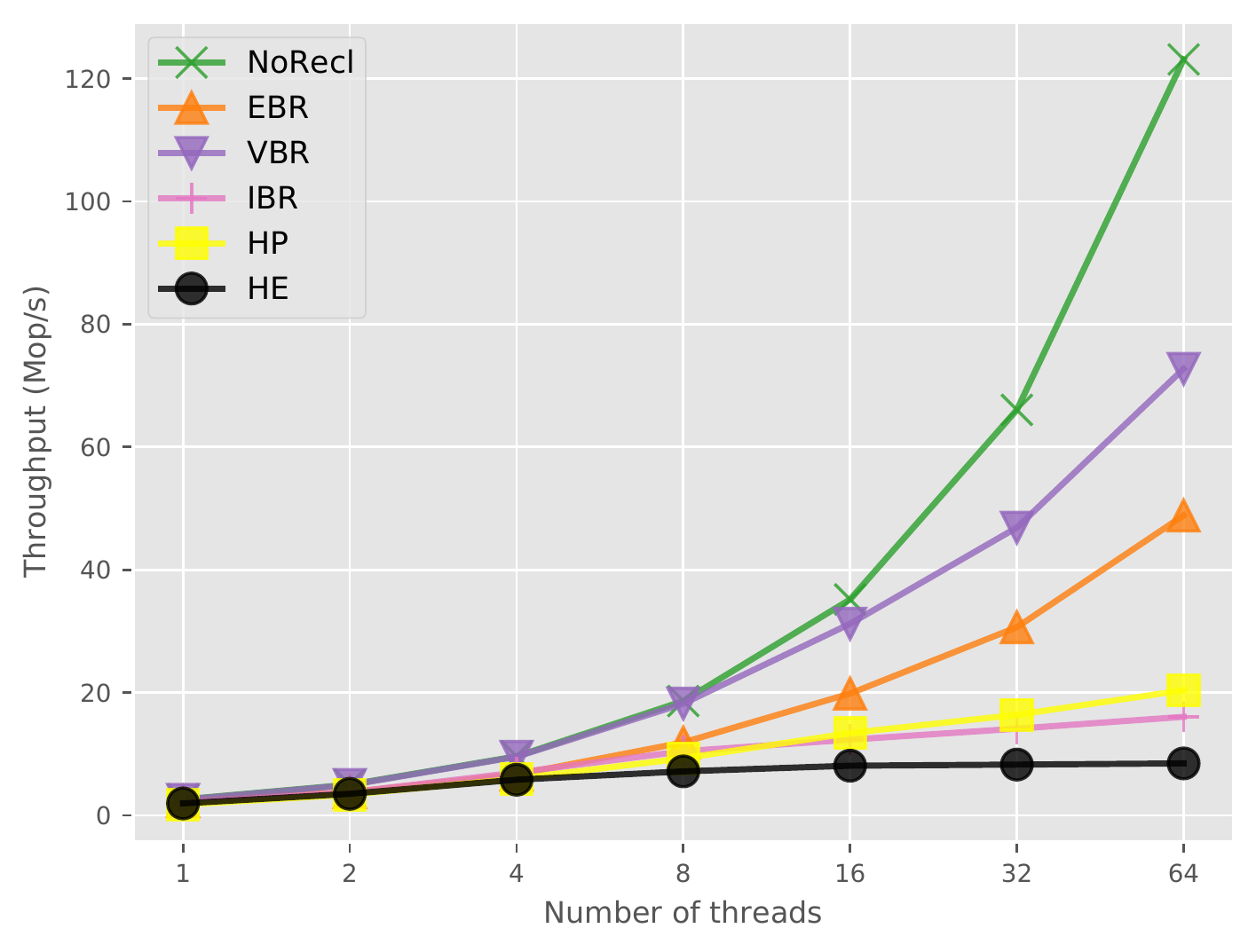}
        \caption{Hash table. Key range: 10M. 25\% inserts 25\% deletes 50\% reads.}
        	\label{fig:hashmap-10m-50}
      \end{subfigure}
\hfill 
      \begin{subfigure}{0.325\textwidth}
      \includegraphics[width=\textwidth]{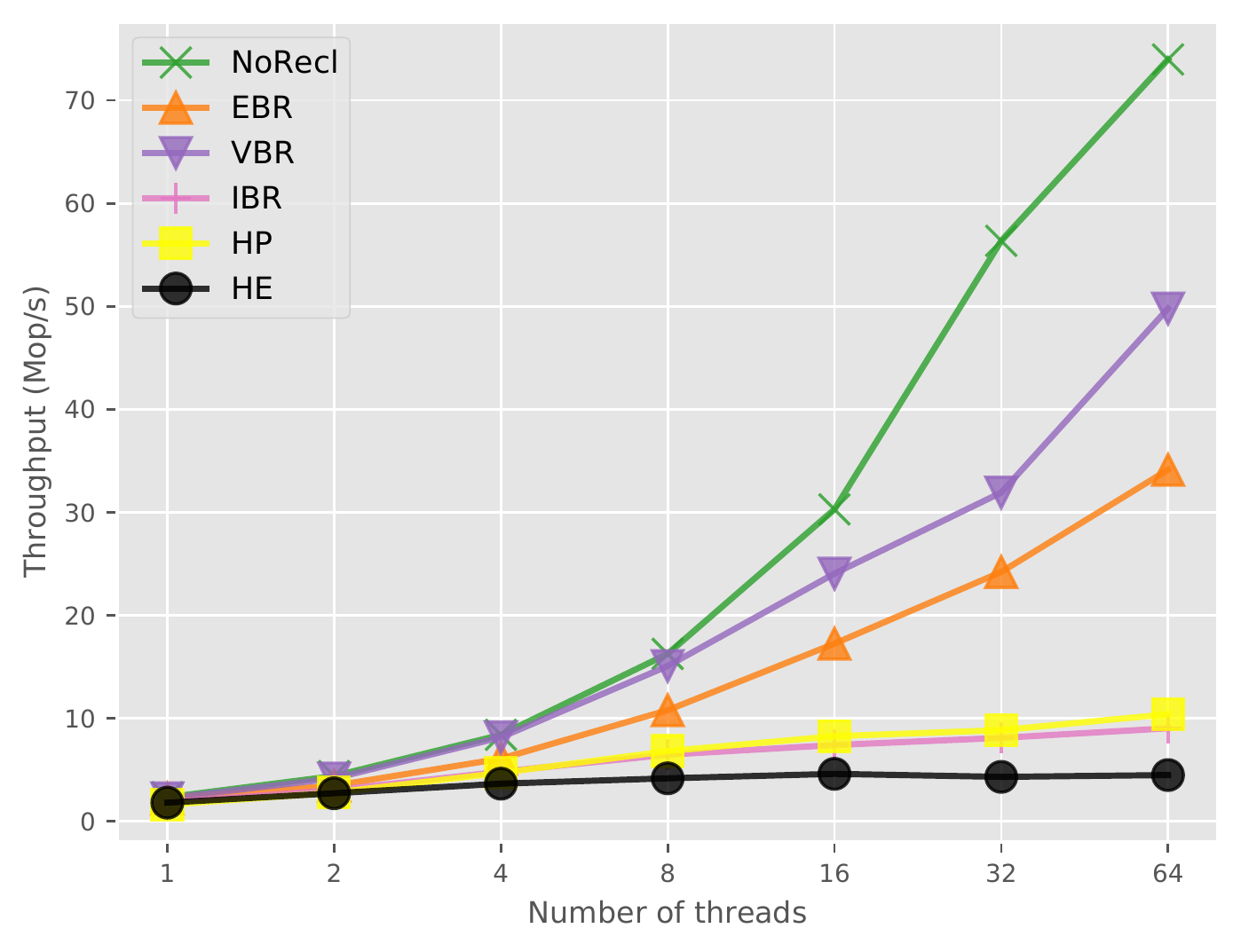}
	  \caption{Hash table. Key range: 10M. 50\% inserts 50\% deletes.}
	  	\label{fig:hashmap-10m-100}
      \end{subfigure}
      
\caption{Throughput evaluation. Y axis: throughput in million operations per second. X axis: \#threads.} \label{fig:throughput}
\end{figure*}

\section{Evaluation} \label{sec-evaluation}

For evaluating throughput of VBR we implemented lock-free variants of a linked-list from~\cite{michael2002high}, a hash table (implemented using the same list), and a skip list. 
We implemented Herlihy and Shavit's lock-free skiplist~\cite{herlihy2020art} with the amendment suggested in~\cite{fraser2004practical} for lock-free reclamation. 
VBR was integrated into all data-structures according to the guidelines presented in Section~\ref{sec-algorithm}.

We evaluated VBR against a baseline execution in which memory is never reclaimed (denoted NoRecl), an optimized implementation of the epoch-based reclamation method~\cite{harris2001pragmatic} (denoted EBR), the traditional hazard pointers scheme~\cite{michael2004hazard} (denoted HP), the hazard eras scheme~\cite{ramalhete2017brief} (denoted HE), and the 2GEIBR variant of the interval-based scheme~\cite{wen2018interval} (denoted IBR).
For all reclamation schemes, we implemented optimized local allocation pools.
Objects were reclaimed once the retire list is full, and were allocated from the shared pool if there were no objects available in the local pool.
As retired objects cannot be automatically reclaimed in EBR, IBR, HE and HP, we tuned their retire list sizes in order to achieve high performance. 
We further tuned the global epoch update rate in EBR, HE and IBR (in VBR it seldom happens and does not require any tuning).

Pointer-based methods require that it would not be possible to reach a reclaimed node by traversing the data structure from a protected node, even if the protected node has been unlinked and retired. 
This prevents schemes like HP, HE, IBR, etc. from being used with some data structures such as Harris's original linked-list~\cite{harris2001pragmatic} or the lock-free binary tree of~\cite{natarajan2014fast,clements2012scalable}. We did not implement binary search trees, because some of the measured competing schemes  cannot support it.

\subsection{Setup} \label{sec-setup}

Our experimental evaluation was performed on an Ubuntu 14.04 (kernel version 4.15.0) OS. The machine featured 4 AMD Opteron(TM) 6376 2.3GHz processors, each with 16 cores (64 threads overall). The machine used 128GB RAM, an L1 cache of 16KB per core, an L2 cache of 2MB for every two cores and an L3 cache of 6MB per processor. The code was compiled using the GCC compiler version 7.5.0 with the -O3 optimization flag. 

We implemented object pools in a similar way to~\cite{cohen2015automatic}, to avoid returning reclaimed objects to the OS. All schemes used that implementation, in which all pools are pre-allocated before the test.
Each test was a fixed-time micro benchmark in which threads randomly call the {\em Insert()}, {\em Delete()} and {\em Search()} operations according to three workload profiles: (1) a search-intensive workload (80\% searches, 10\% inserts and 10\% deletes), (2) a balanced workload (50\% searches, 25\% inserts and 25\% deletes), and (3) an update-intensive workload (50\% inserts and 50\% deletes).
Each execution started by filling the data-structure to half of its range size. 
For the hash-table, the load factor was 1.
We measured the throughput of the above schemes. Each experiment lasted 1 second (as longer executions showed similar results) and was run with a varying number of executing threads. Each experiment was executed 10 times, and the average throughput across all executions was calculated.

\subsection{Discussion} \label{sec-results-discussion}

Figure~\ref{fig:throughput} shows that VBR is faster than other manual reclamation schemes, even when contention is high (Figures~\ref{fig:list-256-20}-\ref{fig:list-256-100}), and in update-intensive workloads (Figures~\ref{fig:list-256-100},~\ref{fig:skiplist-10k-100},~\ref{fig:hashmap-10m-100}).
VBR outperforms epoch-based competitors (EBR, IBR and HE) due to its infrequent epoch updates. 
In order to avoid allocation bottlenecks,  
EBR, IBR and HE require frequent epoch updates. I.e., many global epoch accesses result in cache misses and slow down the allocation process, the reclamation process and the operations executions.
In contrast to these methods, VBR requires infrequent epoch updates. An increment is triggered when the next node to be allocated has a retire epoch equal to the current global epoch.
Most global epoch accesses during VBR hit the cache, and are negligible in terms of performance.
In addition, VBR outperforms its pointer-based competitors (IBR, HE and HP) since it requires neither read nor write fences. 
Specifically, in the hash table implementation, it surpasses the next best algorithm, EBR, by up to 60\% in the search-intensive workload (Figure~\ref{fig:hashmap-10m-20}), by up to 50\% in the balanced workload (Figure~\ref{fig:hashmap-10m-50}), and by up to 40\% in the update-intensive workload (Figure~\ref{fig:hashmap-10m-100}).
In the skiplist implementation, VBR is comparable to the baseline and EBR for the search-intensive and balanced workloads (Figures~\ref{fig:skiplist-10k-20}-\ref{fig:skiplist-10k-50}). For the update-intensive workload, it outperforms the next best algorithm, IBR, by up to 35\% (Figure~\ref{fig:skiplist-10k-100}).
In the linked-list implementation, VBR outperforms the next best algorithm, EBR, by up to 10\%, 11\% and 8\%, respectively (Figures~\ref{fig:list-256-20}-\ref{fig:list-256-100}).
For cache locality reasons, VBR outperforms the baseline execution for all linked-list and skiplist workloads and for all key ranges (Figures~\ref{fig:list-256-20}-\ref{fig:skiplist-10k-100}). As cache locality plays no role in the hash table implementation, VBR has no advantage against the baseline for this data-structure. VBR's throughput is around 75\% of the baseline for the search-intensive workload, around 60\% of the baseline for the balanced workload and around 65\% of the baseline for the update-intensive workload.
\section{Related Work} \label{sec-related}

Much related work was already discussed in the introduction or in the evaluation section (Section~\ref{sec-evaluation}). There are many memory management schemes, and in the evaluation we compared VBR against highly efficient schemes whose code is available (We could not compare against all).  
Previous works~\cite{kang2020marriage,singh2021nbr} defined a set of desirable reclamation properties. Safe reclamation algorithms should be fast (show low latency and high throughput), robust (the number of unreclaimed objects should be bounded), widely applicable and self-contained (not relying on external features).

Two novel methods initiated the study of memory reclamation for concurrent data structures. Pointer-based schemes~\cite{dice2016fast,michael2004hazard,herlihy2005nonblocking} protect objects that are currently accessed by placing a hazard pointer referencing them. These methods are often slow and not always applicable. Epoch-based schemes~\cite{harris2001pragmatic,brown2015reclaiming} (and quiescent state based schemes~\cite{hart2007performance}) wait until all threads move to the next operation to make sure that an unlinked node cannot be further accessed. Such methods  
are sometimes not robust, and most hybrids of the two approaches~\cite{nikolaev2020universal,wen2018interval,ramalhete2017brief,brown2015reclaiming,balmau2016fast} are either not always applicable, or rely on special hardware support.
Drop-the-anchor~\cite{braginsky2013drop} extends HP by protecting only some of the traversed nodes and reclaiming carefully, yet it is not easily applicable and has only been applied to linked-lists.

Another approach, which is neither fast nor robust, is reference counting based reclamation~\cite{blelloch2020concurrent,detlefs2002lock,gidenstam2008efficient,herlihy2005nonblocking}. This scheme keeps an explicitly count of the number of pointers to each object, and reclaims an object with a zero count. Such schemes require a way to break cyclic structures of retired objects, and are often slow or rely on hardware assumptions. This scheme has a wait-free (and in particular, robust) variant~\cite{sundell2005wait} and a lock-free variant~\cite{correia2021orcgc}, but they are not fast, since they require multiple expensive synchronization fences.

Many schemes rely on Hardware-specific or OS features, or 
affect the execution environment. ThreadScan~\cite{alistarh2018threadscan}, StackTrack~\cite{alistarh2014stacktrack}, and Dragojevi{\'c} et al.~\cite{dragojevic2011power} rely on  transactional memory for the reclamation, which is not always available in hardware or may be slow in a software implementation. 
DEBRA+~\cite{brown2015reclaiming} and NBR~\cite{singh2021nbr} use OS signals in order to wake unresponsive threads and allow lock-free progress even for EBR-based methods, if the OS signal implementration is lock-free. Morrison and Afek~\cite{morrison2015temporally} avoid memory fences by waiting for a short while. This relies on specific hardware properties that might not always be available. 
Dice et. al.~\cite{dice2016fast} and PEBR~\cite{kang2020marriage} avoid costly fences by relying on the existence of process-wide memory fences. QSense~\cite{balmau2016fast} requires control of the OS scheduler. In particular, to make hazard pointers visible, threads are periodically swapped out.


\bibliography{refs}

\appendix

\section{Correctness} \label{sec-correctness}

Safe memory reclamation often requires that before reclaiming a retired node, the reclaimer ensures that no other threads hold local pointers to this retired node, as they may be later dereferenced.
However, VBR does not return freed space to the operating system, and this requirement can be relaxed.
In this section we prove Theorem~\ref{theorem-maintain}. I.e., we prove that VBR 
maintains the original program's linearizability and lock-freedom guarantee. 
In Section~\ref{sec-linearizability}, we prove that linearizability is maintained (see Lemma~\ref{lemma-ds-linearizable}), and in Section~\ref{sec-lock-freedom}, we prove that lock-freedom is maintained as well (see Lemma~\ref{lemma-ds-lock-free}).

Our correctness proof relates to applying the interface introduced in Section~\ref{sec-algorithm}. Consequently, it applies to specific data-structures. However, it can be easily shown that the same invariants hold for every lock-free linearizable data-structure implementation, for which our assumptions hold (see Section~\ref{sec-assumptions}).

\subsection{VBR Maintains Linearizability} \label{sec-linearizability}

In this section we prove Lemma~\ref{lemma-ds-linearizable}:

\begin{lemma} \label{lemma-ds-linearizable}
Given a lock-free linearizable data-structure implementation, that satisfy all of the assumptions presented in Section~\ref{sec-assumptions}, the implementation remains linearizable after integrating it with VBR, according to the modifications described in Sections~\ref{sec-allocator}-\ref{sec-modifications}.
\end{lemma}

Our motivation is that linearizability is a {\em local} property~\cite{herlihy1990linearizability}. In our context, this means the following: Given a linearizable data-structure implementation (i.e., all possible executions are represented by linearizable respective histories), if each atomic instruction is replaced by its linearizable implementation,
then the overall data-structure implementation remains linearizable. 

Using this concept, in Section~\ref{sec-linearization-points} we prove some basic claims regarding the effect of each method from Figure~\ref{fig:interface}. We do not set actual linearization points, since methods that  terminate by a rollback step actually do not take effect at all.
In Section~\ref{sec-reclamation-free-executions} we show that integrating a linearizable reclamation-free implementation with checkpoints and rollback steps maintains linearizability.
Using the claims proved in Sections~\ref{sec-linearization-points}-\ref{sec-reclamation-free-executions}, in Sections~\ref{sec-inverse} we prove that every VBR-integrated execution has the same history as some linearizable reclamation-free execution, by constructing the respective execution inductively. As every VBR-integrated execution has a linearizable history, Lemma~\ref{lemma-ds-linearizable} derives.

\subsubsection{Basic Linearizability-Related Claims} \label{sec-linearization-points}

We first prove that our allocation mechanism is safe. I.e., a node is never allocated from a memory address before the previous node, allocated from the same address, is retired.
In order to show that our allocation mechanism is safe, we prove Claim~\ref{claim-alloc-safe}. We first make some basic observations:

\begin{observation} \label{observation-e-increase}
The global epoch $e$ never decreases.
\end{observation}

\begin{observation} \label{observation-node-epoch-changes}
The node's birth and retire epochs only change during allocation and retirement.
\end{observation}

\begin{claim} \label{claim-alloc-safe}
Let $E$ be a VBR-integrated execution, and suppose that two different {\em alloc()} calls return nodes $n_1, n_2$, allocated from the same memory address. W.l.o.g., assume that $n_1$ is returned first. Then:

\begin{enumerate}
    \item If $n_1$ never becomes reachable, then it is appended to an allocation list before $n_2$ is allocated.
    \item If $n_1$ ever becomes reachable, then $n_1$ is retired exactly once.
    \item If $n_1$ ever becomes reachable, then $n_1$'s retirement completes before $n_2$ is allocated.
    \item $n_1$'s birth and retire epochs are smaller than $n_2$'s birth epoch.
\end{enumerate}
\end{claim}

\begin{claimproof}
W.l.o.g., let's assume that $n_2$ is the first node allocated from the memory address $n_1$ has been previously allocated from. By transitivity, the claim holds for any future allocation from this address. 

\begin{enumerate}
    \item Assume that $n_1$ never becomes reachable. By Assumption~\ref{assumption-life}, it is never retired and therefore, is never appended to a retired nodes list. Since $n_2$ is allocated from the same address, $n_1$ must have been re-appended to an allocation list (for more details, see Appendix~\ref{sec-checkpoint-extra}), before $n_2$ is allocated.

    \item Since the rolling-back mechanism may result in calling the {\em retire()} method more than once per node, extra precautions are added. After a new node is allocated, its retire epoch is initialized to $\bot$ (see line~\ref{alloc-update-retire}). A node is retired (i.e., added to a thread's retired nodes list in line~\ref{retire-retire}) only if its retire epoch contains $\bot$ and its birth epoch is the one received as input (otherwise, the {\em retire()} call returns in line~\ref{retire-if-retired}), and it does not contain $\bot$ after being retired (see line~\ref{retire-update}). 
    In addition, according to Assumption~\ref{assumption-life}, a single thread is in charge of retiring every node. Notice that, as can be shown inductively, $n_1$ cannot be accidentally retired by a thread trying to retire a node previously allocated from the same memory address. Therefore, $n_1$ is retired at most once.
    
    \item When $n_2$ is allocated, it is popped out of the thread's allocation list (line~\ref{alloc-retrieve}). Therefore, the previous node, allocated from the same memory address (i.e., $n_1$) must have been previously retired and added to this allocation list.
    
    \item If $n_1$ is never retired, then it is appended to an allocation list (for more details, see Appendix~\ref{sec-checkpoint-extra}), right before a rollback step. Since its birth epoch is set before this rollback step, by Observation~\ref{observation-node-epoch-changes}, its birth epoch is at most the epoch value, as saved by the executing thread. Since this value is updated after rolling-back (see Section~\ref{sec-checkpoints}), $n_2$, allocated by the same thread, has a bigger birth epoch. Since $n_1$'s retire epoch is $\bot$ in this case, $n_2$'s birth epoch is bigger than $n_1$'s retire epoch as well.
    
    Otherwise, if $n_1$ is retired, then $n_1$'s retire epoch is set during its retirement (in line~\ref{retire-update}), which completes before $n_2$'s allocation. When $n_2$ is allocated, if $n_1$'s retire epoch equals the thread's copy of the global epoch, then $n_2$ is not allocated (see lines~\ref{alloc-if}-~\ref{alloc-return-cp}). Since $n_2$'s birth epoch is set to the global epoch only after making sure that it is no longer equal to $n_1$'s retire epoch, $n_2$'s birth epoch is strictly bigger than $n_1$'s retire epoch (and by transitivity, bigger than its birth epoch).
\end{enumerate}

There is one extra case that still needs to be handled. 
Suppose that $n_1$ is retired right before the rollback step (see Appendix~\ref{sec-checkpoint-extra}). In this case, like in a standard retirement, if $n_1$'s retire epoch is $\bot$, then the retiring thread sets its retire epoch to be the current global epoch and appends it to its retired nodes list. Therefore, this case falls back to the standard case, and is already handled above.

\end{claimproof}

Given Claim~\ref{claim-alloc-safe}, we are guaranteed that the life periods of any two nodes, allocated from the same memory address, do not overlap. 
Next, we are going to prove some claims for each method from Figure~\ref{fig:interface}. This claims will later be used in Section~\ref{sec-inverse}, when showing that any VBR-integrated execution has the same history as a linearizable reclamation-free one.

\subparagraph{The {\em isMarked()} method}

Our {\em isMarked()} method always returns a correct answer, as it is not affected by global epoch changes, and never performs rollback steps. This concept is captured in Claim~\ref{claim-is-marked-correct}, that can actually be used as a linearization proof for this method.

\begin{claim} \label{claim-is-marked-correct}
Assume that the input parameters to an {\em isMarked()} call by a thread T are a pointer to a node $n$, allocated from a memory address $a$, and its birth epoch $b$. Then:
\begin{enumerate}
    \item If the method returns TRUE in line~\ref{is-valid-return-false}, then $n$ is already marked at this point.
    \item If the method returns in line~\ref{is-valid-return-res}, then it returns TRUE iff $n$ is marked when its {\em next} pointer is read in line~\ref{is-valid-is-valid}.
\end{enumerate}
\end{claim}

\begin{claimproof}
Recall that $b$ is $n$'s birth epoch, as originally read by T. 
If the method returns in line~\ref{is-valid-return-false}, then by Observation~\ref{observation-node-epoch-changes}, $n$ is already retired when the birth epoch is read in line~\ref{is-valid-return-false}. In particular, by Assumption~\ref{assumption-life}, $n$ is considered as marked (i.e., invalid) at this point, and the method indeed returns TRUE.

Otherwise, if the method returns in line~\ref{is-valid-return-res}, then by Claim~\ref{claim-alloc-safe}, $n$ is not reclaimed before executing line~\ref{is-valid-return-false}. Therefore, the method returns TRUE iff the {\em isMarked()} call in line~\ref{is-valid-is-valid} returns TRUE. I.e., it returns TRUE iff the {\em next} field, read in line~\ref{is-valid-is-valid}, is marked, and the claim holds.
\end{claimproof}

Recall that part~\ref{node-invalid} of Assumption~\ref{assumption-life} holds for every mutable field.
Therefore, although Claim~\ref{claim-is-marked-correct} relates to the {\em next} field specifically, it holds for any mutable node field.

\subparagraph{The {\em getKey()} method}

Our {\em getKey()} method may result in a rollback step to the last checkpoint. Therefore, we only prove that it returns the correct answer when it returns in line~\ref{get-key-return}.
Recall that, as we describe in Section~\ref{sec-read}, this method is always called before the next checkpoint installation. Therefore, it is safe to assume that the executing thread does not update its local copy of the global epoch between installing the pointer and calling the {\em getKey()} method.

\begin{claim} \label{claim-get-key-correct}
Assume that the input parameter to a {\em getKey()} call by a thread T is a pointer to a node $n$. Then the key returned in line~\ref{get-key-return} is equal to $n$'s key.
\end{claim}

\begin{claimproof}
As described in Section~\ref{sec-read}, this method is called after T installs a pointer to $n$, and before installing the next checkpoint. 
Recall that the global epoch is always read in line~\ref{read-if}, after installing pointers, and if it is different than the local copy, then a rollback step is performed. Therefore, In particular, it is guaranteed that T's local copy of the global epoch counter has not updated between installing the local pointer to $n$ and calling the {\em getKey()} method. 
Since the method does not return in line~\ref{get-key-return-cp}, it is guaranteed that the global epoch has not changed until executing line~\ref{get-key-return-cp} as well. 
By Claim~\ref{claim-alloc-safe}, $n$ was not retired before T executed line~\ref{get-key-return-cp}. Therefore, it was not retired when reading $n$'s key in line~\ref{get-key-get-key}, and the claim holds.
\end{claimproof}

Recall that we treat any immutable field in the same way. I.e., threads must read it before installing the next checkpoint. Therefore, although Claim~\ref{claim-get-key-correct} relates to the {\em getKey()} method specifically, it holds for any read of an immutable node field.

\subparagraph{The {\em getNext()} method}

We now turn to handle the {\em getNext()} method.
Our {\em getNext()} method may result in a rollback step to the last checkpoint. Therefore, we only prove that it returns the correct answer when it returns in line~\ref{get-ref-return}, and when the output is indeed used by the calling thread. For more details, see Assumption~\ref{assumption-mutable}.



\begin{claim} \label{claim-get-next-correct}
Assume that the input parameter to a {\em getNext()} call by a thread T is a pointer to a node $n$, allocated from a memory address $a$. Then either (1) the output from the {\em getNext()} call is not used by T, or (2) the output is a pointer to a node $m$, which is $n$'s successor when line~\ref{get-ref-get-ref} is executed, together with $m$'s birth epoch. In particular, $m$ is not yet reclaimed when line~\ref{get-ref-get-ref} is executed.
\end{claim}

\begin{claimproof}
Assume that so far, all {\em getNext()} calls satisfy the claim conditions.
Let $e_1$ be the global epoch value, right after T installs a local pointer to $n$. It is guaranteed that T's local copy of the global epoch is equal to $e_1$ at this point. Otherwise, it must perform a rollback step and discard its pointer to $n$, whether in line~\ref{alloc-return-cp} or~\ref{read-return-cp}.

If the global epoch is no longer $e_1$ when T executes line~\ref{get-ref-get-ref}, then since the method does not return in line~\ref{read-return-cp}, T must have already updated its local copy of the global epoch since installing the local pointer to $n$.
T can only update its local copy of the global epoch after a checkpoint rollback. Therefore, T must have performed a rollback-unsafe step since installing the local pointer to $n$. By Assumption~\ref{assumption-mutable}, there are two possibilities in this case:
\begin{enumerate}
    \item $n$ is already marked at this point, and by Claim~\ref{claim-is-marked-correct}, a call to the {\em isMarked()} method returns the correct answer and T does not use the {\em getNext()} output.
    \item $n$ is not yet marked. By assumption~\ref{assumption-life}, it is still reachable. Thus, $m$ is also reachable when line~\ref{get-ref-get-ref} is executed. In particular, it is not yet retired. Since the global epoch does not change between the executions of lines~\ref{get-ref-get-ref} and~\ref{read-return-cp}, by Claim~\ref{claim-alloc-safe}, $m$ is guaranteed to not be reclaimed before the execution of line~\ref{get_next_birth}. Therefore, by Observation~\ref{observation-node-epoch-changes}, the birth epoch read in line~\ref{get_next_birth} is indeed $m$'s birth epoch.   
\end{enumerate}

It still remains to handle the case in which the global epoch is still $e_1$ when T executes line~\ref{get-ref-get-ref}.
By Claim~\ref{claim-alloc-safe}, $n$ is not yet reclaimed at this point.
Since the method does not return in line~\ref{read-return-cp}, the global epoch is also $e_1$ when T executes line~\ref{get_next_birth}. By Claim~\ref{claim-alloc-safe}, $n$'s successor cannot be reclaimed between the execution of line~\ref{get-ref-get-ref} and the execution of line~\ref{get_next_birth}. By Observation~\ref{observation-node-epoch-changes}, its birth epoch is indeed the birth epoch read in line~\ref{get_next_birth}, which the method eventually returns.
\end{claimproof}

Recall that Assumption~\ref{assumption-invalidate}, Assumption~\ref{assumption-mutable} and Claim~\ref{claim-is-marked-correct} hold for every mutable field. Therefore, although Claim~\ref{claim-get-next-correct} relates to the {\em getNext()} method specifically, it holds for any read of a mutable node field.
Before handling our two update methods, we prove the following claim regarding our versioning mechanism.

\begin{claim} \label{claim-version-correct}
Let $n_1,n_2$ be two nodes, let $b_1,b_2$ be their birth epochs, and let $a_1,a_2$ be the memory addresses they were allocated from, respectively. In addition, assume that $n_1$'s {\em next} field points to $n_2$, with a version $v$. Then: 
\begin{enumerate}
    \item $v$ is the maximum between $b_1$ and $b_2$.
    \item If either $n_1$ or $n_2$ are retired, then their retire epochs are at least $v$.
\end{enumerate}
\end{claim}

\begin{claimproof}
Assume by contradiction that at some point during the execution, the claim does not hold for the first time. Then it must happen upon an update of a node's {\em next} pointer.
A node's {\em next} pointer is only updated in line~\ref{alloc-cas},~\ref{cas-wcas} and~\ref{invalidate-wcas}.
In line~\ref{alloc-cas}, the node's pointer is set to NULL and therefore, the claim still vacuously holds.

Assume that the update is executed in line~\ref{cas-wcas}. Let $n_1,b_1,n_2,b_2,n_3,b_3$ be the {\em update()} method input parameters, respectively. In addition, for every $i \in \{1,2,3\}$, let $a_i$ be the memory address $n_i$ is allocated from.
By the claim assumption, if $n_1$ points to $n_2$, then the pointer's version is the maximum between $b_1$ and $b_2$, as calculated in line~\ref{cas-max-exp}.
If $n_1$ is already retired when line~\ref{cas-wcas} is executed, then by Claim~\ref{claim-alloc-safe}, the birth epoch of the node, allocated from $a_1$ when the WCAS is executed in line~\ref{cas-wcas}, is bigger than $n_1$'s retire epoch. By the claim assumption, it is bigger then the maximum between $b_1$ and $b_2$.
Since the WCAS is successful, the version calculated in line~\ref{cas-max-exp} is indeed the right version. I.e., $n_1$ is not yet retired when the WCAS is executed. By Observation~\ref{observation-node-epoch-changes}, its birth epoch is still $b_1$ at this point.
Moreover, as mentioned in Section~\ref{sec-write}, the WCAS success implies that $n_1$ is not marked, and by Assumption~\ref{assumption-life}, it is reachable. Since $n_3$ is also reachable after the successful WCAS, by Assumption~\ref{assumption-life}, it cannot be retired at this point. Therefore, by Observation~\ref{observation-node-epoch-changes}, its birth epoch is still $b_3$. Therefore:
\begin{enumerate}
    \item After the successful WCAS, $n_1$ points to $n_3$ with a version that is the maximum between $b_1$ and $b_3$, as calculated in line~\ref{cas-max-birth}.
    \item Since both $n_1$ and $n_3$ are not yet retired, by Observation~\ref{observation-e-increase}, their retire epochs are going to be at least the current epoch, which is at least the maximum between $b_1$ and $b_3$.
\end{enumerate}

Therefore, the claim necessarily does not hold for the first time, after an execution of line~\ref{invalidate-wcas}. Let $n_1$ and $b_1$ be the two input parameters to the {\em mark()} call, let $n_2$ be the node read from $n_1$'s {\em next} pointer in line~\ref{mark-read-next}, and let $b_2$ be its birth epoch, read in line~\ref{mark-max-exp}. 
Since the method does not return in line~\ref{mark-return-false}, by Observation~\ref{observation-node-epoch-changes}, $n_1$'s birth epoch is necessarily still $b_1$ when line~\ref{mark-read-next} is executed. By the claim assumption, $n_1$ points to $n_2$ with a version which is the maximum between $b_1$ and $b_2$, as calculated in line~\ref{mark-max-exp}.

Since the WCAS is successful, the version calculated in line~\ref{mark-max-exp} is indeed the right version. I.e., $n_1$ is not yet retired when the WCAS is executed. By Observation~\ref{observation-node-epoch-changes}, its birth epoch is still $b_1$ at this point.
Moreover, the WCAS success implies that $n_1$ is not marked (as the expected value, calculated in line~\ref{mark-read-next}, is not marked), and by Assumption~\ref{assumption-life}, it is reachable. Since $n_2$ is also reachable after the successful WCAS, by Assumption~\ref{assumption-life}, it cannot be retired at this point. Therefore, by Observation~\ref{observation-node-epoch-changes}, its birth epoch is still $b_2$. Therefore:
\begin{enumerate}
    \item After the successful WCAS, $n_1$ points to $n_2$ with a version that is the maximum between $b_1$ and $b_2$, as calculated in line~\ref{mark-max-exp}.
    \item Since both $n_1$ and $n_2$ are not yet retired, by Observation~\ref{observation-e-increase}, their retire epochs are going to be at least the current epoch, which is at least the maximum between $b_1$ and $b_2$.
\end{enumerate}

Since both claim assumptions still hold, we derive a contradiction. I.e., the claim assumptions always hold.
\end{claimproof}

\subparagraph{The {\em update()} method}

\begin{claim} \label{claim-update-precondition}
Assume that the input parameters to an {\em update()} call by a thread T are pointers to the nodes $n_1,n_2,n_3$, allocated from the memory addresses $a_1,a_2,a_3$, and their birth epochs, $b_1,b_2,b_3$, respectively.
The WCAS executed in line~\ref{cas-wcas} is successful iff $n_1$'s next pointer points to $n_2$ with an unmarked pointer, right before executing the CAS.
\end{claim}

\begin{claimproof}
Recall that all input parameters are assumed to be unmarked.
First, assume that the WCAS is successful. Then a node allocated from $a_1$ points to a node allocated from $a_2$, with a version which is the maximum between $b_1$ and $b_2$, as calculated in line~\ref{cas-max-exp}. 
Assume by contradiction that either $n_1$ or $n_2$ are retired at this stage. By Claim~\ref{claim-version-correct}, the retire epoch of each one of them is at least the value calculated in line~\ref{cas-max-exp}. By Claim~\ref{claim-alloc-safe}, the birth epoch of a new node, allocated either from $a_1$ or $a_2$, must be strictly bigger than this value. By Claim~\ref{claim-version-correct}, in this case, the pointer version must be bigger then the value calculated in line~\ref{cas-max-exp} -- a contradiction. Therefore, if the WCAS is successful then $n_1$'s next pointer points to $n_2$ right before the linearization point.

Now, assume that the WCAS is unsuccessful, and assume by contradiction that $n_1$ points to $n_2$ with an unmarked pointer.
By Claim~\ref{claim-version-correct}, the pointer's version must be the one calculated in line~\ref{cas-max-exp} -- a contradiction to the WCAS failure, and the claim follows.
\end{claimproof}

\begin{claim} \label{claim-update-postcondition}
Assume that the input parameters to an {\em update()} call by a thread T are pointers to the nodes $n_1,n_2,n_3$, allocated from the memory addresses $a_1,a_2,a_3$, and their birth epochs, $b_1,b_2,b_3$, respectively.
If the WCAS executed in line~\ref{cas-wcas} is successful, then after the CAS, $n_1$'s next pointer points to $n_3$ with an unmarked pointer.
\end{claim}

\begin{claimproof}
Assume that the WCAS executed in line~\ref{cas-wcas} is successful.
Recall that all input node pointers are assumed to be received as unmarked pointers.
Since the WCAS is successful, after it is executed, a node allocated from $a_1$ points (with an unmarked node) to a node allocated from $a_3$, with a version which is the maximum between $b_1$ and $b_3$, as calculated in line~\ref{cas-max-birth}. 
Assume by contradiction that either $n_1$ or $n_3$ are retired at this stage. By Claim~\ref{claim-version-correct}, the retire epoch of each one of them is at least the value calculated in line~\ref{cas-max-birth}. By Claim~\ref{claim-alloc-safe}, the birth epoch of a new node, allocated either from $a_1$ or $a_3$, must be strictly bigger than this value. By Claim~\ref{claim-version-correct}, in this case, the pointer version must be bigger then the value calculated in line~\ref{cas-max-birth} -- a contradiction. Therefore, if the WCAS is successful, then $n_1$'s next pointer points to $n_3$ right after the linearization point.
\end{claimproof}

Recall that Assumption~\ref{assumption-invalidate} holds for every mutable field. Moreover, the version of non-pointer mutable fields should always be equal to the node's birth epoch. Therefore, Claims~\ref{claim-update-precondition} and~\ref{claim-update-postcondition} can be easily adjusted to fit other mutable fields. 

\subparagraph{The {\em mark()} method}

Let $n_1,b_1$ be the two input parameters to a {\em mark()} call, executed by a thread T. I.e., $b_1$ is $n_1$'s birth epoch, as previously saved by T. Let $n_2$ be the node $n_1$ points to when its {\em next} pointer is read in line~\ref{mark-read-next}, and let $b_2$ be the birth epoch, read in line~\ref{mark-max-exp}.
In addition, let $a_1,a_2$ be the addresses $n_1$ and $n_2$ are allocated from, respectively. If the method returns in line~\ref{mark-return-false}, then its linearization point is set to be the read of the birth epoch in line~\ref{mark-return-false}.
Otherwise, if the method returns in line~\ref{invalidate-wcas}, then its linearization point is set to be the WCAS execution in line~\ref{invalidate-wcas}.
We prove that the method indeed takes effect in its linearization point using the following claims:

\begin{claim} \label{claim-mark-return-false}
Assume that the input parameters to a {\em mark()} call by a thread T are a pointer to a node $n_1$, allocated from the memory addresses $a_1$, and its birth epoch, $b_1$.
If the method returns FALSE in line~\ref{mark-return-false} then $n_1$ is already marked when this line is executed.
\end{claim}

\begin{claimproof}
By Observation~\ref{observation-node-epoch-changes}, $n_1$ is already retired when the birth epoch is read in line~\ref{mark-return-false}. In particular, by Assumption~\ref{assumption-life}, $n_1$ is considered as marked (i.e., invalid) at this point.
\end{claimproof}

\begin{claim} \label{claim-mark-preconditions}
Assume that the input parameters to a {\em mark()} call by a thread T are a pointer to a node $n_1$, allocated from the memory addresses $a_1$, and its birth epoch, $b_1$.
Let $n_2$ be $n_1$'s successor when $n_1$'s {\em next} pointer is read in line~\ref{mark-read-next}.
The WCAS executed in line~\ref{invalidate-wcas} is successful iff prior to its execution, $n_1$ points to $n_2$ with an unmarked pointer.
\end{claim}

\begin{claimproof}
Assume that the WCAS executed in line~\ref{invalidate-wcas} is successful. Since the expected pointer is unmarked (as calculated in line~\ref{mark-read-next}), it remains to show that $n_1$ indeed points to $n_2$.
Assume by contradiction that either $n_1$ or $n_2$ are retired at this stage. By Claim~\ref{claim-version-correct}, the retire epoch of each one of them is at least the value calculated in line~\ref{mark-max-exp}. By Claim~\ref{claim-alloc-safe}, the birth epoch of a new node, allocated either from $a_1$ or $a_2$, must be strictly bigger than this value. By Claim~\ref{claim-version-correct}, in this case, the pointer version must be bigger then the value calculated in line~\ref{mark-max-exp} -- a contradiction. Therefore, if the WCAS is successful then $n_1$'s next pointer points to $n_2$ right before to the WCAS.

Now, assume that the WCAS executed in line~\ref{invalidate-wcas} is unsuccessful. If either $n_1$'s {\em next} pointer is marked, or the node allocated from $a_1$ does not point to the node allocated from $a_2$ at this stage, then we are done.
Otherwise, assume by contradiction that $n_1$ points to $n_2$ with an unmarked pointer. By Claim~\ref{claim-version-correct}, the pointer's version must be the one calculated in line~\ref{mark-max-exp} -- a contradiction to the WCAS failure, and the claim follows.
\end{claimproof}

\begin{claim} \label{claim-mark-postconditions}
Assume that the input parameters to a {\em mark()} call by a thread T are a pointer to a node $n_1$, allocated from the memory addresses $a_1$, and its birth epoch, $b_1$.
Let $n_2$ be $n_1$'s successor when $n_1$'s {\em next} pointer is read in line~\ref{mark-read-next}.
If the WCAS executed in line~\ref{invalidate-wcas} is successful, then $n_1$ points to $n_2$ via a marked pointer, right after the WCAS execution.
\end{claim}

\begin{claimproof}
If the WCAS is successful, then the updated pointer is indeed marked (as the new pointer value is marked in line~\ref{mark}).
Since the WCAS is successful, after it is executed, a node allocated from $a_1$ points to a node allocated from $a_2$, with a version which is the maximum between $b_1$ and $b_2$, as calculated in line~\ref{mark-max-exp}. 
Assume by contradiction that either $n_1$ or $n_2$ are retired at this stage. By Claim~\ref{claim-version-correct}, the retire epoch of each one of them is at least the value calculated in line~\ref{mark-max-exp}. By Claim~\ref{claim-alloc-safe}, the birth epoch of a new node, allocated either from $a_1$ or $a_2$, must be strictly bigger than this value. By Claim~\ref{claim-version-correct}, in this case, the pointer version must be bigger than the value calculated in line~\ref{mark-max-exp} -- a contradiction. Therefore, if the WCAS is successful, then $n_1$'s next pointer points to $n_2$ right after the linearization point.
\end{claimproof}

Recall that Assumption~\ref{assumption-invalidate} hold for every mutable field. Moreover, the version of non-pointer mutable fields should always be equal to the node's birth epoch. Therefore, Claims~\ref{claim-mark-return-false},~\ref{claim-mark-preconditions} and~\ref{claim-mark-postconditions} can be easily adjusted to fit other mutable fields. 

\subsubsection{Inserting Checkpoints into Reclamation-Free Implementations} \label{sec-reclamation-free-executions}

We first show that, given any (and in particular, a reclamation-free) execution, checkpoint rollbacks preserve its set of possible history extensions. 
Recall that rollback-safe steps are defined in Definition~\ref{definition-rollback-safe}, and may be either local steps, shared-memory reads or shared memory writes. In addition, given a linearizable implementation, installing checkpoints is well-defined in Section~\ref{sec-checkpoints}.
After  a thread installs a checkpoint, a rollback to this checkpoint includes restoring the local variables, saved upon installing this checkpoint. Intuitively, by Definition~\ref{definition-rollback-safe}, threads execute only rollback-safe steps between checkpoints and thus, a rollback to the last checkpoint is always safe. We prove this notion in the following claim:

\begin{claim} \label{claim-compose}
Assume that $s_j$ is a rollback-safe step, executed by a thread T. Let $0<i<j$, and assume that for every $i<t < j$, if $s_t$ is a step executed by T, then $s_t$ is also a rollback-safe step.
Let $E'$ be the execution obtained by removing from $E_j$ all steps $s_t$ (for every $i<t \leq j$), executed by T during $E_j$. Then EXT($E'$)=EXT($E_j$).
\end{claim}

\begin{claimproof}
We are going to prove the claim by induction on $j-i$. For the base case, the claim holds by Definition~\ref{definition-rollback-safe}.
For the induction step, assume that the claim holds for $j-1$. Let $E''$ be the execution obtained by removing all steps $s_t$ (for every $i<t \leq j-1$), executed by T during $E_{j-1}$. Then $E''=E'$, and by the induction hypothesis, $EXT(E'')=EXT(E_{j-1})$. I.e., $EXT(E')=EXT(E_{j-1})$, and by Definition~\ref{definition-rollback-safe}, $EXT(E')=EXT(E_j)$, and the claim holds.
\end{claimproof}

Our next goal is to show equivalence between any given linearizable implementation and the implementation obtained after inserting checkpoints and rollback instructions into the given code.
Let $E$ be a linearizable execution, and assume that $s_j$ is a rollback step by a thread T. Since checkpoints are installed upon every operation invocation, there exists at least one step $s_i$ (for some $i<j$) which is T's last checkpoint visit. W.l.o.g., let $i$ be the maximal such index before $j$. Then by Claim~\ref{claim-compose}, executing $s_j$ maintains the set of possible history extensions. In particular, since the original execution is linearizable, then integrating it with checkpoints and rollback steps maintains its linearizability. We conclude with the next Corollary to Claim~\ref{claim-compose}, using its transitivity property:

\begin{corollary} \label{corollary-rollback-maintain}
Given any linearizable implementation that satisfy the assumptions from Section~\ref{sec-assumptions}, integrating it with checkpoints according to the guidelines from Section~\ref{sec-checkpoints}, and any number of rollback instructions, maintains its linearizability.
\end{corollary}
\subsubsection{The Inverse Transformation} \label{sec-inverse}

Let $I^{RF}$ be a linearizable reclamation-free implementation, satisfying the assumptions from Section~\ref{sec-assumptions}. By Corollary~\ref{corollary-rollback-maintain}, we can assume that it also contains checkpoint instructions, and that its linearizability property is not affected by inserting rollback steps.
Let $I$ be the implementation obtained after applying the transformation from Section~\ref{sec-algorithm} to $I^{RF}$.
Given a VBR-integrated execution $E \in I$, we construct a reclamation-free execution $E^{RF} \in I^{RF}$ by induction on $E$'s length\footnote{We ignore {\em retire()} calls, as they do not affect linearizability}.

First, we map node addresses as follows: if a node is allocated during $E$ from a memory address $a$ and with a birth epoch $b$, then during $E^{RF}$, it is mapped to a node, allocated from the memory address $\langle a,b\rangle$. This mapping is legal since the memory is considered to be unbounded in the reclamation-free setting. Moreover, by Claim~\ref{claim-alloc-safe}, it is safe to assume that every two different nodes are allocated from different addresses in the reclamation-free setting.

In addition, any local variable from the original implementation $I^{RF}$ (not added after the VBR integration) is mapped to itself. Now, consider the finite VBR-integrated execution $E_{i-1} \in I$ and its respective reclamation-free execution $E_{i-1}^{RF} \in I^{RF}$, obtained so far.

\begin{enumerate}
    \item If $s_i$ is an operation invocation, an operation response, or a local step from the original execution (not added after the VBR integration), then it is also appended to $E_{i-1}^{RF}$.
    
    \item Assume that $s_i$ is a return from an {\em alloc($k$)} call in line~\ref{alloc-return}, let $a$ be the address $n$ is allocated from, and let $b$ be its birth epoch (as set in line~\ref{alloc-update-birth}). Then a respective allocation of a node from the memory address $\langle a,b\rangle$, by the same thread, and an initialization of its key to be $k$, are appended to $E_{i-1}^{RF}$.
    
    
    
    \item Assume that $s_i$ is the read of a node's {\em next} pointer in line~\ref{get-ref-get-ref}, by a thread T. Assume that the input parameter to this {\em getNext()} call is a pointer to a node, allocated from a memory address $a$ with a birth epoch $b$. Let $m$ be the node , allocated from $\langle a,b\rangle$ during $E_{i-1}^{RF}$.
    If T returns in line\ref{get-ref-return}, then a read of $m$'s {\em next} pointer and a respective {\em unmark()} call, both by T, are appended to $E_{i-1}^{RF}$.
    
    \item Assume that $s_i$ is the read of a node's key in line~\ref{get-key-get-key}, by a thread T. Assume that the input parameter to this {\em getKey()} call is a pointer to a node, allocated from a memory address $a$ with a birth epoch $b$. Let $m$ be the node , allocated from $\langle a,b\rangle$ during $E_{i-1}^{RF}$.
    If T returns in line\ref{get-key-return}, then a read of $m$'s key is appended to $E_{i-1}^{RF}$.
    
    \item Assume that $s_i$ is the read of a node's key, after this key has been previously saved in a local pointer by a thread T (for more details, see Section~\ref{sec-read}). Assume that this node is allocated from the address $a$ and with a birth epoch $b$. Then a read of the key of the node, allocated from $\langle a,b\rangle$, is appended to $E_{i-1}^{RF}$.
    
    \item Assume that $s_i$ is the read of a node's {\em next} pointer in line~\ref{is-valid-is-valid}, by a thread T. Assume that the input parameter to this {\em isMarked()} call is a pointer to a node, allocated from a memory address $a$ with a birth epoch $b$. Let $m$ be the node , allocated from $\langle a,b\rangle$ during $E_{i-1}^{RF}$.
    If T returns in line~\ref{is-valid-return-res}, then a respective {\em isMarked()} call is appended to $E_{i-1}^{RF}$.
    
    \item Assume that $s_i$ is returning TRUE in line~\ref{is-valid-return-false}, by a thread T. Assume that the input parameter to this {\em isMarked()} call is a pointer to a node, allocated from a memory address $a$ with a birth epoch $b$. Let $m$ be the node , allocated from $\langle a,b\rangle$ during $E_{i-1}^{RF}$.
    Then a read of $m$'s {\em next} pointer and a respective {\em isMarked()} call, both by T, are appended to $E_{i-1}^{RF}$.
    \item Assume that $s_i$ is the execution of line~\ref{cas-wcas}. Let $a_1,b_1,a_2,b_2,a_3,b_3$ be the input parameters to the respective {\em update()} call. For every $i \in \{1,2,3\}$, let $n_i$ be the node allocated from $\langle a_i,b_i\rangle$ during $E_{i-1}^{RF}$. Then a CAS on $n_1$'s {\em next} field, with $n_2$ as the expected value and $n_3$ as the new one, is appended to $E_{i-1}^{RF}$.
    
    \item Assume that $s_i$ is the read of a node's {\em next} pointer in line~\ref{mark-read-next}, by a thread T. Assume that the input parameter to this {\em mark()} call is a pointer to a node, allocated from a memory address $a$ with a birth epoch $b$. Let $m$ be the node , allocated from $\langle a,b\rangle$ during $E_{i-1}^{RF}$.
    If T returns in line~\ref{invalidate-wcas}, then a read of $m$'s {\em next} pointer into a local variable $exp$ is appended to $E_{i-1}^{RF}$.
    
    \item Assume that $s_i$ is the execution of line~\ref{invalidate-wcas}, by a thread T. Assume that the input parameter to this {\em mark()} call is a pointer to a node, allocated from a memory address $a$ with a birth epoch $b$. Let $m$ be the node , allocated from $\langle a,b\rangle$ during $E_{i-1}^{RF}$, and let $exp$ be the respective local variable from the previous transformation step. Then a mark of $m$'s {\em next} field, with $exp$ as the expected value, is appended to $E_{i-1}^{RF}$.
    
    \item Assume that $s_i$ is returning FALSE in line~\ref{mark-return-false}, by a thread T. Assume that the input parameter to this {\em mark()} call is a pointer to a node, allocated from a memory address $a$ with a birth epoch $b$. Let $m$ be the node , allocated from $\langle a,b\rangle$ during $E_{i-1}^{RF}$.
    Then a mark attempt of $m$'s {\em next} pointer is appended to $E_{i-1}^{RF}$.
    
    \item If $s_i$ is a checkpoint installation, then a checkpoint installation by the same thread is appended to $E_{i-1}^{RF}$.
    
    \item If $s_i$ is a rollback step, then a rollback step by the same thread is appended to $E_{i-1}^{RF}$.
\end{enumerate}

Other steps are discarded when constructing $E^{RF}$.
Recall that successful shared memory updates are only executed in lines~\ref{cas-wcas} and~\ref{mark}, which are the respective last instructions in both update methods. Therefore, checkpoint instructions are installed in the same places in both implementation.
Next, we prove the following equivalence between $E_i$ and its respective transformation output $E_i^{RF}$:

\begin{claim} \label{claim-legal-execution}
For every $i>0$, $E_i^{RF} \in I^{RF}$.
\end{claim}

\begin{proof}
We prove the claim by induction on $i$. For the induction step, assume that $E_{i-1}^{RF} \in I^{RF}$.
Note that we do not claim for correctness at this stage. We just prove that the execution $E^{RF}$ follows the original reclamation-free code.
\begin{itemize}
    \item If $s_i$ is an operation invocation, an operation response, or a local step from the original execution (not added after the VBR integration), then it is also appended to $E_{i-1}^{RF}$, and obviously, $E_i^{RF} \in I^{RF}$.
    
    \item Assume that $s_i$ is a return from an {\em alloc()} call in line~\ref{alloc-return}. Then a respective allocation appeared in $I_{RF}$. Since the return from the {\em alloc()} call is the only step resulting in appending an allocation to $E_{i-1}^{RF}$, the claim still holds for $E_i^{RF}$ in this case.

    

    \item Assume that $s_i$ is the read of a node's {\em next} pointer in line~\ref{get-ref-get-ref}, by a thread T. Then a respective read of this {\em next} pointer, along with an {\em unmark} call, appeared in $I_{RF}$. Since the read of line~\ref{get-ref-get-ref} is the only step resulting in appending the respective read and unmark of a {\em next} pointer to $E_{i-1}^{RF}$, the claim still holds for $E_i^{RF}$.
    
    \item Assume that $s_i$ is the read of a node's key in line~\ref{get-key-get-key}. Then a respective read of this key appeared in $I_{RF}$. Since the read of line~\ref{get-key-get-key} is the only step resulting in appending the respective read key to $E_{i-1}^{RF}$, the claim still holds for $E_i^{RF}$.
    
    \item Assume that $s_i$ is the read of a node's key, after it has been previously saved in a local variable. Then a respective read of this key appeared in $I_{RF}$, and the claim still holds for $E_i^{RF}$.

    \item Assume that $s_i$ is the read of a node's {\em next} pointer in line~\ref{is-valid-is-valid}, by a thread T, that returns in line~\ref{is-valid-return-res}. Then a respective {\em isMarked()} call appeared in $I_{RF}$. Since a respective {\em isMarked()} call is appended to $E_{i-1}^{RF}$ only with respect to the execution of line~\ref{is-valid-is-valid} in this case, the claim still holds for $E_i^{RF}$.

    \item Assume that $s_i$ is returning TRUE in line~\ref{is-valid-return-false}, by a thread T. Then a respective {\em isMarked()} call appeared in $I_{RF}$. Since a respective {\em isMarked()} call is appended to $E_{i-1}^{RF}$ only with respect to the execution of line~\ref{is-valid-return-false} in this case, the claim still holds for $E_i^{RF}$.

    \item Assume that $s_i$ is the execution of line~\ref{cas-wcas}. Then a respective pointer update appeared in $I_{RF}$. Since the WCAS execution in line~\ref{cas-wcas} is the only step resulting in appending the respective update to $E_{i-1}^{RF}$, the claim still holds for $E_i^{RF}$.
    
    \item Assume that $s_i$ is the read of a node's {\em next} pointer in line~\ref{mark-read-next}, by a thread T. Then a respective pointer read, in the scope of a {\em mark()} execution, appeared in $I_{RF}$. Since the execution of line~\ref{mark-read-next} is the only step resulting in appending the respective pointer read to $E_{i-1}^{RF}$, the claim still holds for $E_i^{RF}$.

    \item Assume that $s_i$ is the execution of line~\ref{invalidate-wcas}, by a thread T. Then a respective pointer mark, in the scope of a {\em mark()} execution, appeared in $I_{RF}$. Since the execution of line~\ref{invalidate-wcas} is the only step resulting in appending the respective marking to $E_{i-1}^{RF}$, the claim still holds for $E_i^{RF}$.
    
    \item Assume that $s_i$ is returning FALSE in line~\ref{mark-return-false}, by a thread T. Then a respective pointer mark, in the scope of a {\em mark()} execution, appeared in $I_{RF}$. Since the execution of line~\ref{mark-return-false} is the only step resulting in appending the respective marking to $E_{i-1}^{RF}$, the claim still holds for $E_i^{RF}$.
    
    \item Assume that $s_i$ is a checkpoint installation by a thread T. Then the previous step by T was either an operation invocation or a shared-memory update in line~\ref{cas-wcas} or~\ref{invalidate-wcas}. I.e, by the induction hypothesis, the last step in $E_{i-1}^{RF}$ it the execution of a checkpoint-trigger as well, and the claim still holds for $E_i^{RF}$.
    
    \item Assume that $s_i$ is a rollback step. By Corollary~\ref{corollary-rollback-maintain}, rollback steps can be installed at any stage, and the claim still holds for $E_i^{RF}$.
\end{itemize}
\end{proof}

After proving that $E^{RF}$ indeed follows the original code of the reclamation-free implementation (including checkpoints and rollback steps). It still remains to show that $E$ and $E^{RF}$ share the same history, as we prove in Claim~\ref{claim-equivalent} below.

\begin{claim} \label{claim-equivalent}

Let $i>0$. Then the following hold:
\begin{enumerate}
    \item \label{induction-shared} After $E_i$ and after $E_i^{RF}$, the shared memory view is identical.
    \item \label{induction-local} If a local variable $x$ contains a value $v$ after $E_i^{RF}$, then it either contains $v$ after $E_i$, or its content is indistinguishable to the executing thread from $\bot$ after $E_i$.
    \item \label{induction-history} $E_i$ and $E_i^{RF}$ share the same history.
\end{enumerate}
\end{claim}

\begin{claimproof}
We are going to prove by induction that if the invariants hold for $E_{i-1}$ and $E_{i-1}^{RF}$, then they hold for $E_i$ and $E_i^{RF}$ as well.
For the inductive step, assume that the claim holds for $E_{i-1}$ and $E_{i-1}^{RF}$.
\begin{enumerate}
    \item If $s_i$ is either discarded or has no influence on shared memory, then the claim still holds by the induction hypothesis (invariant~\ref{induction-shared}). Otherwise, it must either be the execution of line~\ref{cas-wcas} or~\ref{mark}. By the induction hypothesis (invariants~\ref{induction-shared} and~\ref{induction-local}), and by Claims~\ref{claim-update-precondition}-\ref{claim-mark-postconditions}, the claim still holds for $E_i$ and $E_i^{RF}$.
    
    \item If $s_i$ is discarded, then the claim still holds by the induction hypothesis (invariant~\ref{induction-local}). 
    Otherwise, if $s_i$ is an operation invocation, an operation response, or a local step from the original execution (not added after the VBR integration), then the claim still holds by the induction hypothesis (invariants~\ref{induction-shared} and~\ref{induction-local}). Otherwise:
    
    \begin{itemize}
        \item If $s_i$ is the return from an {\em alloc()}, resulting in the read of the allocated node into a thread's local pointer, then by our transformation, the respective node, with the same key, is read into the thread's same local variable in $E_i^{RF}$, and the claim still holds.
        \item If $s_i$ is the read of a node's key into a local variable, then by the induction hypothesis (invariant~\ref{induction-shared}) and  Claim~\ref{claim-get-key-correct}, the claim still holds.
        \item If $s_i$ is the read of a node's {\em next} field into a local variable (that is not ignored by the executing thread), then by Assumption~\ref{assumption-mutable}, the induction hypothesis (invariant~\ref{induction-shared}), Claim~\ref{claim-get-next-correct} and Claim~\ref{claim-mark-preconditions}, then the claim still holds.
        If the executing thread ignores the read value, then 
        it is indistinguishable to it from $\bot$, and the claim still holds.
        \item If $s_i$ is an {\em isMarked()} call then by the induction hypothesis (invariant~\ref{induction-shared}) and  Claim~\ref{claim-is-marked-correct}, the claim still holds.
        \item If $s_i$ is a read of a pointer in line~\ref{mark-read-next}, then 
        it is indistinguishable to the executing thread from $\bot$, and the claim still holds.
        \item If $s_i$ is a pointer update, then its response is read into a local variable. By the induction hypothesis (invariants~\ref{induction-shared} and~\ref{induction-local}) and Claims~\ref{claim-update-precondition}-\ref{claim-mark-postconditions}, the claim still holds.
        \item If $s_i$ is a checkpoint installation then the same respective local variables are logged. By the induction hypothesis (invariants~\ref{induction-shared} and~\ref{induction-local}), the claim still holds.
        \item If $s_i$ is a rollback step, then by the induction hypothesis (invariants~\ref{induction-shared} and~\ref{induction-local}), both executions rollback to the same checkpoint and restore the same local variables, and the claim still holds.
    \end{itemize}
    
    \item If $s_i$ is neither an invocation nor a response step, then the history of both executions remains unchanged. If $s_i$ is an invocation step, then by our construction, an identical invocation is appended to $E_{i-1}^{RF}$. If $s_i$ is a response step, then by our construction, a respective response step is appended to $E_{i-1}^{RF}$. By invariant~\ref{induction-local}, the response output is identical in both executions. Finally, by the induction hypothesis (invariant~\ref{induction-history}), the claim still holds.
   
\end{enumerate}
\end{claimproof}

By Claims~\ref{claim-legal-execution} and~\ref{claim-equivalent}, every VBR-integrated execution $E \in I$ has the same history as a certain reclamation-free execution $E^{RF} \in I^{RF}$. By our assumption, $I^{RF}$ contains only linearizable executions, which derives Lemma~\ref{lemma-ds-linearizable}.

\subsection{VBR Maintains Lock-Freedom} \label{sec-lock-freedom}

In this section we prove that VBR maintains the lock-freedom guarantee of the original linearizable and lock-free implementation:

\begin{lemma} \label{lemma-ds-lock-free}
Given a lock-free linearizable data-structure implementation, that satisfy all of the assumptions presented in Section~\ref{sec-assumptions}, the implementation remains lock-free after integrating it with VBR, according to the modifications described in Sections~\ref{sec-allocator}-\ref{sec-modifications}.
\end{lemma}

We first prove that VBR is {\em robust}. I.e., a
stalled thread may not prevent the reclamation of an unbounded number of retired objects and therefore, the system never blocks due to a possible exhaustion of the heap. 
We start by showing that the number of unused nodes that are not retired is bounded:

\begin{claim} \label{claim-recycled}
Let $E$ be a VBR-integrated execution and let $C_i$ be a configuration. The number of unlinked nodes (see Assumption~\ref{assumption-life}) that are not retired at $C_i$ is bounded.
\end{claim}

\begin{claimproof}
During $E$, an unlinked node $n$ must be in one of the following statuses:

\begin{enumerate}
    \item $n$ was allocated by a thread $T$, and has not become reachable since, and $T$ is now performing a rollback step. After $T$ performs its rollback step, no thread will have a local pointer to $n$, and by Assumption~\ref{assumption-life}, $n$ will also not be reachable from shared-memory.
    \item $n$ has gone through stages~\ref{node-allocated}-\ref{node-unlinked} of Assumption~\ref{assumption-life}, and has not been retired yet.
\end{enumerate}

By Assumption~\ref{assumption-life}, there exist no other possibilities. 
The first case is handled as stated in Appendix~\ref{sec-checkpoint-extra}. Before $T$ rolls back, it appends $n$ to its allocation list. Although it is not a standard retirement, it makes sure that $n$ is recycled.

The second case is naturally handled, according to Assumption~\ref{assumption-life}. If the retiring thread is given with enough scheduler time, then it should eventually retire $n$. If it is not given enough time, then the number of unlinked nodes that are still not retired is still bounded.

In addition, as described in Appendix~\ref{sec-checkpoint-extra}, if $T$ is forced to execute a roll-back step before it gets the chance to retire $n$, then it retires $n$ right before rolling back (i.e., if $n$'s retire epoch is $\bot$, it sets its retire epoch to be the current global epoch and appends it to its retired nodes list).
Notice that if $n$'s retirement depends on following steps (i.e., there exists an extension of the current execution in which $T$ does not retire $n$), then $n$ should not be retired at this stage anyway.
\end{claimproof}

After showing that the number of unused nodes that are not retired is bounded, we show that the number of retired nodes that are not reclaimed is also bounded.
Recall that the executing threads use local retired nodes list. Therefore, a stalled thread may prevent the reclamation of the nodes residing in its own retired nodes list. One can easily add a {\em stealing} mechanism, allowing other threads to reclaim nodes that were retired by other threads. However, we have not implemented such a mechanism, since the size of these lists is bounded (and pre-defined), and does not affect the system progress in practice. 
For showing that VBR is robust, it still remains to show that a stalled thread does not affect the reclamation of nodes, retired by other threads. We prove a stronger claim:

\begin{claim} \label{claim-robust}
Let $E$ be a VBR-integrated execution, and let $s_i$ and $s_j$ be two consecutive calls to the {\em alloc()} method by a thread T. Then at least one call returns a new allocated node in line~\ref{alloc-return}.
\end{claim}

\begin{claimproof}
If the first call to the {\em alloc()} method returns in line~\ref{alloc-return}, then we are done. Otherwise, it returns in line~\ref{alloc-return-cp}, after T tries to increase the global epoch in line~\ref{alloc-update-e}, and after the new node, $n$, is appended to the allocation list in line~\ref{alloc-realloc}. 

When the {\em alloc()} method is called again by T, 
If the node retrieved from the allocation list is not $n$, then it must be a node $m$, returned to T's allocation list before performing the rollback step. The node $m$ is guaranteed to not be reachable yet (for more details, see Section~\ref{sec-checkpoints}). By Assumption~\ref{assumption-life}, $m$ is also guaranteed to not be retired. Therefore, $m$'s retire epoch is $\bot$, the condition in line~\ref{alloc-if} does not hold and $m$ is returned in line~\ref{alloc-return}.

Otherwise, $n$ is  node retrieved from the allocation list in line~\ref{alloc-retrieve}.
By assumption~\ref{assumption-life}, no other thread tries to meanwhile retire $n$, and thus, its retire epoch does not change before T calls the {\em alloc()} method for the second time.
By Observation~\ref{observation-e-increase}, the global epoch can only increase. Therefore, even if T's CAS in line~\ref{alloc-update-e} (during the first {\em alloc()} call) is not successful, the global epoch is guaranteed to be bigger than $n$'s retire epoch when T rolls-back in line~\ref{alloc-return-cp}. After executing line~\ref{alloc-return-cp}, T updates its local copy of the global epoch. Consequently, it is guaranteed to be bigger than $n$'s retire epoch during the second {\em alloc()} call, and the second call returns in line~\ref{alloc-return}.
\end{claimproof}

After showing that VBR is robust, it still remains to show that it does not foil the original implementation's lock-freedom guaranty.
To derive a contradiction, assume that there exists a VBR-integrated execution $E$ that does not satisfy the lock-freedom guaranty. Thus there is a suffix of $E$ in which no operations terminate and some operations take infinitely
many steps. 
Recall that there is a finite number of executing threads, and that we assume a well-formed execution. Therefore, at every given point during $E$, there is a finite number of pending invocations.

\begin{claim} \label{claim-inf-logical-insertions}
There is a finite number of logical node insertions, logical node removals and node retirements during $E$.
\end{claim}

\begin{claimproof}
Assume by contradiction that there are infinitely many logical node insertions during $E$. 
In particular, there is a pending operation $op$ during which there are infinitely many logical node insertions.
By Assumption~\ref{assumption-life}, only previously allocated nodes can be logically inserted into the data-structure. Therefore, there exist at least two different points after the invocation of $op$, that are considered as logical insertions by the thread executing $op$. 
By Lemma~\ref{lemma-ds-linearizable}, $E$ is linearizable.
By the {\em non-blocking} property of linearizability (see~\cite{herlihy1990linearizability}), $op$ has two different linearization points -- a contradiction. Therefore, there is a finite number of logical node insertions during $E$.

By Assumption~\ref{assumption-life}, a node can only be logically removed after it is logically inserted, and can only be retired after it is logically removed. 
In addition, by Claim~\ref{claim-alloc-safe}, every node is retired exactly once.
Therefore, there is also a finite number of logical node removals and retirements during $E$.
\end{claimproof}

\begin{claim} \label{claim-inf-epoch-changes}
There is a suffix of $E$ in which the global epoch counter does not change.
\end{claim}

\begin{claimproof}
By Claim~\ref{claim-inf-logical-insertions}, there is a finite number of node retirements during $E$. Therefore, there is a suffix $\beta$ of $E$ in which node retirements are not executed anymore.
Let $e$ be the retire epoch of the last retired node during $E$. By Observation~\ref{observation-e-increase}, the retire epoch of every previously retired node is at most $e$. 

The global epoch counter can only be increased in line~\ref{alloc-update-e}, and only if it equal to the retire epoch of a newly allocated node. During $\beta$, all nodes have a retire epoch which is either at most $e$, or $\bot$. Therefore, the global epoch does not increase after it is incremented to $e+1$.
\end{claimproof}

\begin{claim} \label{claim-inf-rollbacks}
There is a suffix of $E$ in which rollback steps are never executed.
\end{claim}

\begin{claimproof}
By Claim~\ref{claim-inf-epoch-changes}, there is a suffix $\alpha$ of $E$ in which the global epoch counter does not change.
Let $T$ be a thread, executing a rollback step during $\alpha$. Upon rolling-back to its last checkpoint, $T$ updates its local copy of the global epoch counter. Starting from this update, $T$'s local copy of the global epoch counter remains equal to the global epoch counter, and $T$ never execute a rollback step again.
The number of rollback steps during $\alpha$ is bounded by the number of executing threads (which is finite), and the claim holds.
\end{claimproof}

Let $s_i$ be the first step in the suffix of $E$, guaranteed by Claim~\ref{claim-inf-rollbacks}. By Claim~\ref{claim-legal-execution}, the reclamation-free execution $E^{RF}$, obtained by our transformation, is a legal execution of the original reclamation-free implementation.
Starting from the step which is equivalent to $s_i$, $E^{RF}$ contains no rollback steps as well. Therefore, by the lock-freedom guaranty of the original implementation, some pending operation terminates during $E^{RF}$. By Claim~\ref{claim-equivalent}, it terminates during $E$ as well -- a contradiction.
VBR does not foil the original implementation's lock-freedom guaranty, which derives Lemma~\ref{lemma-ds-lock-free}.

\section{Handling Unlinked Nodes before Checkpoint Rollbacks} \label{sec-checkpoint-extra}

In this section we explain how we avoid losing access to unlinked nodes when rolling back to a checkpoint. This obviously does not affect linearizability, as by Assumption~\ref{assumption-life}, such nodes cannot represent data-structure items anyway.
Moreover, this does not directly affect robustness, as robustness relates to the number of unreclaimed retired nodes, which is bounded in VBR (for more details, see Appendix~\ref{sec-lock-freedom}). 
However, our goal is to show that the system never blocks as a result of an allocation block (due to an exhaustion of the heap). Therefore, we handle such unlinked nodes as follows.

Assume that a thread T is about to execute a rollback step to its previous checkpoint, Then there are two types of nodes that may become unreachable to all of the executing threads after T executes its rollback step:

\begin{enumerate}
    \item Nodes that were allocated by T after its last checkpoint visit, and are not yet reachable from the data-structure entry points.
    \item Unlinked nodes that should be retired by T.
\end{enumerate}

By Assumption~\ref{assumption-life}, all other nodes are either reachable or must be retired by other threads. The first type of unlinked nodes is handled easily: for every such node, T appends them to its allocation list. By Assumption~\ref{assumption-life}, these nodes are not yet retired and therefore, their retire epoch is $\bot$. In Appendix~\ref{sec-correctness} we show that this does not foil correctness.

Handling the other type is more complex. Recall that by Assumption~\ref{assumption-life}, such nodes should already be unreachable. Therefore, T has direct access to them via its local pointers, or by following a sequence of {\em next} pointers, starting from a node referenced by one of its local pointers (e.g., retiring a sequence of unlinked nodes after a physical deletion in~\cite{harris2001pragmatic}).
By Assumption~\ref{assumption-life}, T should be able to distinguish between such nodes and nodes that should not be retired.
Additionally, when T is in charge of retiring a certain node, and this node is only reachable by a sequence of node pointers, it is guaranteed that all of the nodes in the sequence are also retired by T. In particular, they are not yet reclaimed, and by Assumption~\ref{assumption-invalidate}, their pointer data does not change throughout this retirement procedure.
Therefore, T can safely read the relevant node pointers, while ignoring pointer versions.
For every such node $n$, if $n$'s retire epoch is $\bot$ then T updates $n$'s retire epoch to be the current global epoch, and appends $n$ to its retired nodes list.
Note that T does not update its local copy of the global epoch, as its is going to be updated after the rollback step.

\section{Integrating VBR: An Example} \label{sec-example}

We are going to demonstrate VBR's integration using a standard implementation of a lock-free linked-list~\cite{harris2001pragmatic,herlihy2020art,michael2002high}. 
The linked-list interface suggests three operations: the {\em add()} operation adds a key to the list (and does nothing if the key is already present), the {\em remove()} operation removes an existing key (and does nothing if no such key is present), and the {\em contains()} operation returns a boolean indicating whether a key is present. In addition to these three operations, the implementation includes an auxiliary {\em find()} method, used by the {\em add()} and {\em remove()} operations for locating a given key. The {\em find()} method is also in charge of physically deleting marked nodes that has not been physically removed by their logical remover yet (logical and physical deletions are discussed in Section~\ref{sec-assumptions}).

\begin{figure}[!ht]
	\begin{algorithmic}[1]\footnotesize
        \State \textbf{find}(int key) \label{contains-call}
		\State retry: \label{find-retry}
		\Indent
	        \State pred := head \label{find-head}
	        \State pred\_b := head $\rightarrow$ birth\_epoch \label{find-head-birth}
            \State $\langle$ curr, curr\_b $\rangle$ := {\color{red}getNext}(pred) \label{find-head-next}
	        \State curr\_key := {\color{red}getKey}(curr) \label{find-head-key}
		    \State while (TRUE) \label{find-while}
		    \Indent
                \State if ({\color{red}isMarked}(curr, curr\_b)) \label{find-is-marked}
                \Indent
                    \State $\langle$ succ, succ\_b $\rangle$ := {\color{red}getNext}(curr) \label{find-curr-next}
                    \State while ({\color{red}isMarked}(succ, succ\_b))  \label{find-while-succ-marked}
                    \Indent
                        \State $\langle$ succ, succ\_b $\rangle$ := {\color{red}getNext}(succ) \label{find-succ-next}   
                    \EndIndent
                    \State if ({\color{red}update}(pred, pred\_b, curr, curr\_b, succ, succ\_b) == FALSE) \label {find-physical-deletion}
                    \Indent
                        \State goto retry \label{find-goto-retry}
                    \EndIndent
                    \State $\langle$ curr, curr\_b $\rangle$ := $\langle$ succ, succ\_b $\rangle$ \label{find-curr-gets-succ}
                    \State curr\_key := {\color{red}getKey}(curr) \label{find-succ-key}
                \EndIndent
                \State if (curr\_key $\geq$ key) \label{find-if-done}
                \Indent
                    \State return pred, pred\_b, curr, curr\_b   \label{find-return} 
                \EndIndent
                \State $\langle$ pred, pred\_b $\rangle$ := $\langle$ curr, curr\_b $\rangle$ \label{find-pred-gets-curr}
                \State $\langle$ curr, curr\_b $\rangle$ := {\color{red}getNext}(pred) \label{find-pred-next}
                \State curr\_key := {\color{red}getKey}(curr) \label{find-curr-key}
    		  \EndIndent
        \EndIndent
	\end{algorithmic}
	\caption{\small The {\em find} method.}
	\label{fig:find}
\end{figure}

\begin{figure}[!ht]
	\begin{algorithmic}[1]\footnotesize
        \State \textbf{add}(int key) \label{insert-call}
		\Indent
		    \State {\color{blue}install checkpoint} \label{insert-cp-1}
	        \State my\_e := e.get() \label{insert-read-e-1}
		    \State while (TRUE) \label{insert-while}
		    \Indent
		        \State $\langle$ pred, pred\_b, succ, succ\_b $\rangle$ := find(key) \label{insert-find}
		        \State succ\_key := {\color{red} getKey}(succ) \label{insert-succ-key}
		        \State if (succ\_key == key) return FALSE \label{insert-return-false}
		        \State n := {\color{red}alloc}(key) \label{insert-alloc}
		        \State n\_b := n $\rightarrow$ birth\_epoch \label{insert-read-birth}
		        \State res := {\color{red}update}(pred, pred\_b, succ, succ\_b, n, n\_b) \label{insert-update}
		        \State if (res == TRUE) \label{insert-if-update}
		        \Indent 
		            \State {\color{blue}install checkpoint} \label{insert-cp-2}
		            \State my\_e := e.get() \label{insert-read-e-2}
		            \State return TRUE \label{insert-return-true}
		        \EndIndent
		        \State else {\color{red}retire}(n, n\_b) \label{insert-retire}
    		  \EndIndent
        \EndIndent
	\end{algorithmic}
	\caption{\small The {\em add} operation.}
	\label{fig:insert}
\end{figure}

\begin{figure}[!ht]
	\begin{algorithmic}[1]\footnotesize
        \State \textbf{remove}(int key) \label{remove-call}
		\Indent
		    \State {\color{blue}install checkpoint} \label{remove-cp-1}
	        \State my\_e := e.get() \label{remove-read-e-1}
	        \State $\langle$ pred, pred\_b, curr, curr\_b $\rangle$ := find(key) \label{remove-find}
	        \State curr\_key := {\color{red}getKey}(curr) \label{remove-curr-key}
	        \State if (curr\_key $\neq$ key) return FALSE \label{remove-if-not-found}
	        \State while ({\color{red}isMarked}(curr, curr\_b) == FALSE) \label{remove-while-not-marked}
	        \Indent 
	            \State $\langle$ succ, succ\_b $\rangle$ := {\color{red}getNext}(curr) \label{remove-curr-next}
	            \State res := {\color{red}mark}(curr, curr\_b) \label{remove-mark}
    	        \State if (res == TRUE) \label{remove-if-mark}
	            \Indent 
	                \State {\color{blue}install checkpoint} \label{remove-cp-2}
	                \State my\_e := e.get() \label{remove-read-e-2}
	                \State if ({\color{red}update}(pred, pred\_b, curr, curr\_b, succ, succ\_b) == FALSE) \label{remove-physical-trial}
	                \Indent
	                    \State find(key) \label{remove-find-physical-removal}
	                \EndIndent
	                \State {\color{red}retire}(curr, curr\_b) \label{remove-retire}
	                \State return TRUE \label{remove-return-true}
	            \EndIndent
	        \EndIndent
	        \State return FALSE \label{remove-return-false}
        \EndIndent
	\end{algorithmic}
	\caption{\small The {\em remove} operation.}
	\label{fig:remove}
\end{figure}

\begin{figure}[!ht]
	\begin{algorithmic}[1]\footnotesize
        \State \textbf{contains}(int key) \label{contains-call}
		\Indent
		    \State {\color{blue}install checkpoint} \label{contains-cp}
	        \State my\_e := e.get() \label{contains-read-e}
	        \State curr := head \label{contains-head}
	        \State curr\_b := head $\rightarrow$ birth\_epoch \label{contains-head-birth}
	        \State curr\_key := {\color{red}getKey}(curr) \label{contains-head-key}
		    \State while (curr\_key < key) \label{contains-while}
		    \Indent
	            \State $\langle$ curr, curr\_b $\rangle$ := {\color{red}getNext}(curr) \label{contains-curr-next}
	            \State curr\_key := getKey(curr) \label{contains-curr-key}
    		  \EndIndent
		    \State return ({\color{red}isMarked}(curr, curr\_b) == FALSE) \label{contains-return}
        \EndIndent
	\end{algorithmic}
	\caption{\small The {\em contains} operation.}
	\label{fig:contains}
\end{figure}

As briefly discussed in Section~\ref{sec-checkpoints}, the standard linked-list implementation consists of several update instructions. Specifically, we chose to rely on an implementation with four update instructions.
The first two are the insertion of a new node into the list (executed during the {\em add()} operation) and the marking of a node for logically deleting it (executed during the {\em remove()} operation). 
In both cases, the successful update determines the linearization point of the respective operation. I.e., the output of the respective CAS execution is crucial for linearizability.
If the executing thread performs a rollback step, right after one of these updates, the obtained execution would have a different set of corresponding history extensions (and in particular, it would not be linearizable).
Therefore, these two updates are considered as rollback-unsafe steps, and when are executed successfully, must be followed by a checkpoint installation.

The remaining two update instructions, which are the physical deletion of a node (either during a {\em remove()} operation or the {\em find()} auxiliary method), are both rollback-safe steps. Intuitively, the remover identity does not affect linearization, so performing a rollback step right after the physical deletion would not change the set of corresponding history extensions.
Therefore, checkpoints are not installed after a successful execution of these update instructions.

The standard lock-free linked-list implementation is a natural candidate for demonstrating the VBR integration. Assumption~\ref{assumption-cas} holds since all shared-memory updates occur via CAS executions. Assumption~\ref{assumption-invalidate} holds since a node's \textit{next} pointer is its only mutable field, and it is indeed invalidated via marking. Assumption~\ref{assumption-life} holds since list nodes indeed follow the life-cycle, presented in Section~\ref{sec-assumptions}: they are first allocated, then become reachable and valid (a node's logical insertion into the list is its physical insertion), then invalid (via marking), physically removed, and retired -- either by their logical remover (after making sure they are physically removed) or by some physical remover.
Finally, Assumption~\ref{assumption-mutable} holds since, after a rollback-unsafe update, the list traversal is never resumed from a previously referenced node. E.g., after marking a node for its logical deletion, the following list traversal (for physically deleting this node) is initiated from the list head, and from any other previously saved node reference.

The pseudo code for the VBR integration example appears in Figures~\ref{fig:find}-\ref{fig:contains}. The methods from Figure~\ref{fig:interface} are marked with red and checkpoint installations are marked with blue.
As described in Section~\ref{sec-checkpoints}, installing a checkpoint includes updating the checkpoint reference. The checkpoint reference is used when a thread performs a rollback step (after performing line~\ref{alloc-return-cp}, \ref{retire-return-cp}, \ref{read-return-cp} or \ref{get-key-return-cp} from Figure~\ref{fig:interface}). It should also include saving the content of local variables for a later recovery. However, this part is unnecessary in the linked-list implementation, as local variables are never overwritten after a checkpoint. Finally, the local copy of the global epoch counter is always updated after installing (or rolling back to) a checkpoint.

\subsection{The {\em find()} method}

The original {\em find()} auxiliary method receives a key as its input parameter, and returns pointers to two nodes, \textit{pred} and \textit{curr}. \textit{pred} represents the node with the maximal key which is smaller than the input key in the list, and \textit{curr} represents the node with the minimal key which is equal or greater than the input key in the list. In addition, it is guaranteed that at some point during the {\em find()} execution, both node were reachable (and in particular, logically in the list), and that \textit{curr} was \textit{pred}'s successor at this point.
This method is also in charge of physically removing logically deleted nodes, while traversing the list.

The VBR-integrated pseudo code for the {\em find()} auxiliary method appears in Figure~\ref{fig:find}.
This variant returns the two respective nodes, together with their birth epochs (see line~\ref{find-return}). 
As this method is not an interface operation, a checkpoint is never installed upon its invocation. Additionally, as all updates performed in its scope (i.e., the physical removal in line~\ref{remove-find-physical-removal}) are rollback-safe steps, it dose not include installing checkpoints at all.
Consequently, whenever a read method (either {\em getKey()} or {\em getNext()}) results in a rollback, the rollback is to a checkpoint, installed during the calling operation (either an {\em add()} or a {\em remove()} operation). 

\subsection{The {\em add()} operation}

The original {\em add()} operation receives a key as its input parameter, and inserts a node with the given key into the list (and returns TRUE upon a successful insertion). If there already exists a node with the given key, which is logically in the list, it does nothing (and returns FALSE). The successful physical insertion of a new node is also its logical insertion, and is done via a CAS instruction. The operation returns TRUE only after such a successful CAS, which means that the CAS result is crucial for maintaining linearizability. Therefore, the physical insertion of a node is considered as a rollback-unsafe step according to Definition~\ref{definition-rollback-safe}. 

The VBR-integrated pseudo code for the {\em add()} operation appears in Figure~\ref{fig:insert}.
As this is an interface operation, a checkpoint is installed immediately after invocation (see line~\ref{insert-cp-1}). In addition, a second checkpoint is installed after the successful physical (and logical) insertion of a new node into the list (see line~\ref{insert-cp-2}). This second checkpoint is obviously unnecessary, as there are no rollback steps between the installation and the return in line~\ref{insert-return-true}. However, we added it for the completeness of our example.

\subsection{The {\em remove()} operation}

The original {\em remove()} operation receives a key as its input parameter, and removes a node with the given key from the list (and returns TRUE upon a successful removal). If there does not exist a node with the given key, which is logically in the list, it does nothing (and returns FALSE).
After traversing the list and locating the deletion candidate, it is first logically deleted via marking (the remover is the thread executing the successful respective CAS), then physically removed and retired.
The physical removal is not considered as a crucial phase of the {\em remove()} operation. In certain implementations, the {\em remove()} operation returns right after the logical deletion, and the physical deletion is executed during a future traversal (and may not happen at all). 
In our example implementation, the removing thread is also in charge of the node's retirement. Therefore, a node's physical removal is validated via an extra call to the {\em find()} method, before retirement.

The VBR-integrated pseudo code for the {\em remove()} operation appears in Figure~\ref{fig:remove}.
As this is an interface operation, a checkpoint is installed immediately after invocation (see line~\ref{remove-cp-1}). In addition, a second checkpoint is installed after the successful marking of the victim node (see line~\ref{remove-cp-2}). This second checkpoint is necessary, as the remover thread may rollback during the {\em find()} execution (line~\ref{remove-find-physical-removal}) or when retiring the victim node (line~\ref{remove-retire}).
Note that the physical removal trial in line~\ref{remove-physical-trial} is considered as a rollback-safe step according to Definition~\ref{definition-rollback-safe} and thus, it is not followed by a checkpoint installation. In addition, note that the {\em getNext()} call in line~\ref{remove-curr-next} does not violate Assumption~\ref{assumption-mutable}, as no rollback-unsafe step is executed prior to this reference read.

\subsection{The {\em contains()} operation}

The original {\em contains()} operation receives a key as its input parameter. It returns TRUE if the list contains a node with the given key, and FALSE otherwise.
While the {\em contains()} operation may be implemented using the {\em find()} auxiliary method, it is traditionally implemented in a wait-free manner, with a single list traversal, while avoiding physical deletions (which may foil wait-freedom).

The VBR-integrated pseudo code for the {\em contains()} operation appears in Figure~\ref{fig:contains}.
Although it is built on the wait-free variant of the original {\em contains()} implementation, it is not wait-free (as checkpoint rollbacks do not maintain wait-freedom). As it contains no shared-memory updates, it includes a single checkpoint installation, upon invocation (line~\ref{contains-cp}).

\end{document}